\documentclass[11pt,a4paper]{article}

\usepackage{hyperref}
\usepackage{color}
\usepackage{amsthm}
\usepackage{mathtools}

\usepackage[normalem]{ulem} 
\usepackage{graphicx}
\usepackage{amssymb,latexsym,cite}
\usepackage{amsmath}
\usepackage{amsfonts}
\usepackage{mathrsfs}
\usepackage{bbm}
\usepackage{bm}
\usepackage[T1]{fontenc}

\usepackage[matrix,arrow,color]{xy}

\usepackage{enumitem}

\usepackage{scalerel,stackengine}
\stackMath
\newcommand\reallywidehat[1]{%
\savestack{\tmpbox}{\stretchto{%
  \scaleto{%
    \scalerel*[\widthof{\ensuremath{#1}}]{\kern-.6pt\bigwedge\kern-.6pt}%
    {\rule[-\textheight/2]{1ex}{\textheight}}
  }{\textheight}%
}{0.5ex}}%
\stackon[1pt]{#1}{\tmpbox}%
}

\newenvironment{myitemize}{\begin{itemize}[itemsep=-0.05cm, leftmargin=*, topsep=0.1cm]}{\end{itemize}}

\def\slasha#1{\setbox0=\hbox{$#1$}#1\hskip-\wd0\hbox to\wd0{\hss\sl/\/\hss}}

\def\periodb#1{\setbox0=\hbox{$#1$}#1\hskip-\wd0\hbox to\wd0{-}}

\newcommand{\unit}{\mathbbm{1}}   			
\newcommand{\id}{\mathrm{id}}   			
\newcommand{\ii}{\mathrm{i}}   			

\newcommand{\CA}{\mathcal{A}}    			

\newcommand{\CCA}{\mathscr{A}}

\newcommand{\CB}{\mathcal{B}}

\newcommand{\CCI}{\mathscr{I}}

\newcommand{\CCF}{\mathscr{F}}

\newcommand{\CH}{\mathcal{H}}
\newcommand{\CCH}{\mathscr{H}}
\newcommand{\CI}{\mathcal{I}}
\newcommand{\CJ}{\mathcal{J}}

\newcommand{\CN}{\mathcal{N}}

\newcommand{\CCP}{\mathscr{P}}

\newcommand{\CCT}{\mathscr{T}}

\newcommand{\CCU}{\mathscr{U}}

\newcommand{\longhookrightarrow}{\lhook\joinrel\longrightarrow}

\newcommand{\frg}{\mathfrak{g}}				

\def\tv{{\textrm{\tiny $V$}}}
\def\tv1{{\textrm{\tiny $V[1]$}}}

\newcommand{\fk}{\mathbbm{k}}     			
\newcommand{\FR}{\mathbbm{R}}     			
\newcommand{\FC}{\mathbbm{C}}     			
\newcommand{\FH}{\mathbbm{H}}     			
\newcommand{\FO}{\mathbbm{O}}     			
\newcommand{\MM}{\mathbbm{M}}     			
\newcommand{\RZ}{\mathbbm{Z}}     			

\newcommand{\dd}{\mathrm{d}}     			

\newcommand{\sEnd}{\mathsf{End}}

\newcommand{\comment}[1]{}     				
\def\tyng(#1){\hbox{\tiny$\yng(#1)$}}			
\def\tyoung(#1){\hbox{\tiny$\young(#1)$}}			

\newcommand{\beq}{\begin{eqnarray}}
\newcommand{\eeq}{\end{eqnarray}}

\newcommand{\obs}{\mathsf{Obs}}

\definecolor{outrageousorange}{rgb}{1.0, 0.43, 0.29}

\newcommand{\ttH}{\mathtt{H}}
\newcommand{\ttS}{\mathtt{S}}

\newcommand{\Tr}{\mathrm{Tr}}

\theoremstyle{definition}
\newtheorem{theorem}[equation]{Theorem}

\newtheorem{lemma}[equation]{Lemma}
\newtheorem{proposition}[equation]{Proposition}

\theoremstyle{definition}
\newtheorem{definition}[equation]{Definition}

\newtheorem{remark}[equation]{Remark}
\newtheorem{example}[equation]{Example}

\def\beq{\begin{equation}}
\def\bee{\begin{equation}}
\def\eeq{\end{equation}}
\def\bea{\begin{eqnarray}}
\def\eea{\end{eqnarray}}
\def\ba{\begin{align}}
\def\ea{\end{align}}

\numberwithin{equation}{section}
\catcode`@=12

\begin{document}

\renewcommand{\thefootnote}{\fnsymbol{footnote}}

\begin{titlepage}
	
	\renewcommand{\thefootnote}{\fnsymbol{footnote}}
	
	\begin{center}
		
		\vspace{2cm}
		
		\baselineskip=24pt
		
		{\Large\bf An algebraic formulation of nonassociative quantum mechanics}
		
		\baselineskip=14pt
		
		\vspace{1cm}
		
{\bf\large Peter Schupp}${}^{\,(a),}$\footnote{Email: \ {\tt  pschupp@constructor.university}  \qquad  https://orcid.org/0000-0002-9159-6086}  \ \ and \ \ {\bf\large Richard
			J. Szabo}${}^{\,(b),}$\footnote{Email: \ {\tt R.J.Szabo@hw.ac.uk} \qquad https://orcid.org/0000-0003-3821-4891}
		\\[6mm]
		
\noindent {${}^{(a)}$ \it School of Science, Constructor University\\ Campus Ring 1, 28759 Bremen, Germany} \\[3mm]
\noindent  {${}^{(b)}$ {\it Department of Mathematics, Heriot-Watt University\\ Colin Maclaurin Building,
			Riccarton, Edinburgh EH14 4AS, U.K.}}\\ and {\it Maxwell Institute for
			Mathematical Sciences, Edinburgh, U.K.} 
		\\[20mm]
		
	\end{center}

	\begin{abstract}
		\noindent
We develop a version of quantum mechanics that can handle nonassociative algebras of observables and which reduces to standard quantum theory in the traditional associative setting. Our algebraic approach is naturally probabilistic and is based on using the universal enveloping algebra of a general nonassociative algebra to introduce a generalized notion of associative composition product. We formulate properties of states together with notions of trace, and use them to develop GNS constructions. We describe Heisenberg and Schr\"odinger pictures of completely positive dynamics, and we illustrate our formalism on the explicit examples of finite-dimensional matrix Jordan algebras as well as the octonion algebra.
	\end{abstract}

{\baselineskip=12pt
	\tableofcontents
}
	\end{titlepage}

\setcounter{footnote}{0}
\renewcommand{\thefootnote}{\arabic{footnote}}


\section{Introduction}
\label{sec:Intro}

Quantum mechanics is a probabilistic theory and so relies heavily on a notion of positivity. States of a quantum system form a convex body, and the theory is characterised by features such as superposition of amplitudes as well as entanglement. Noncommutativity is a key ingredient of the theory, and is measured by the commutator $[A,B] = A\,B-B\,A$ of operators. The noncommutative algebra of operators acts on a Hilbert space of states. Alternatively, in the algebraic approach to quantum theory, operators are encoded into $C^*$-algebras on which states are continuous linear functionals giving the expectation values of measurement outcomes; the Hilbert space picture is then recovered through the Gel'fand-Naimark-Segal (GNS) construction. These approaches both inherently involve associative algebras. 

In this paper we revisit the challenge of incorporating nonassociativity into quantum mechanics, which is measured by the associator $[A,B,C] = (A\,B)\,C - A\,(B\,C)$ of operators. We vastly generalise the treatment initiated in~\cite{Mylonas:2013jha} which developed an extension of the phase space formulation of quantum mechanics for a particular model based on nonassociative deformation quantization; see~\cite{Mylonas:2013odt} for a summary. Though usually deemed not possible, in the following we formulate a complete and proper theory of quantum mechanics for nonassociative algebras of observables. 
We demonstrate how to give them a probabilistic interpretation, treating issues such as probability, positivity and convexity, thus carefully checking that our formulation fits into the framework of \emph{generalized probability theories} (a framework for all physical theories). Our work constitutes a conservative yet mathematically rigorous and model-independent approach to the topic which goes just a little bit beyond traditional quantum mechanics in this scheme.
We show that our framework can handle algebraic expressions and observables that would not be available in an associative setting. 

When is such a nonassociative extension of quantum mechanics important in physics? Generalizations of canonical commutation relations in the traditional approach lead to noncommutative geometry and appear naturally in string theory with flat $B$-fields as well as in quantum mechanics with source-free magnetic fields. For example, the momentum commutator $[p_i,p_j]=\ii\,F_{ij}$ gives rise to the Lorentz force. However, the associative approach in unable to deal with extensions of these physical scenarios to systems that involve magnetic sources. 

The prototypical example is  the quantum mechanics of an electrically charged particle moving in a magnetic monopole background, where the Jacobiator of kinematical momenta
\begin{align}
\mathtt{Jac}(p_ i,p_ j,p_ k) = [p_ i,[p_ j,p_ k]] + [p_ k,[p_ i,p_ j]] + [p_ j,[p_ k,p_ i]]
\end{align}
defines an observable for the magnetic charge-current density $j_{\rm m}$ through
\begin{align} \label{eq:Jacjm}
\tfrac13\,\mathtt{Jac}(p_ i,p_ j,p_ k)\,\dd x^ i\wedge\dd x^ j\wedge\dd x^ k = \dd F = \ast\, j_{\rm m} \ .
\end{align}
The relation \eqref{eq:Jacjm} was originally interpreted by Jackiw~\cite{Jackiw:1984rd}, who argued that the associator for three finite spatial translations generated by the kinematical momenta $p_ i$ measures the phase of the total magnetic flux, i.e. the monopole charge, and in turn the failure of the Dirac quantization condition.

In~\cite{Mylonas:2013jha} the Jacobiator was interpreted via a canonical transformation as giving a volume operator which measures coarse-graining of space in non-geometric $R$-flux backgrounds of string theory. This perspective was extended to certain M-theory $R$-flux backgrounds in~\cite{Kupriyanov:2018xji}, where the canonical transformation leads to a curved momentum space induced by a background of non-geometric Kaluza-Klein monopoles~\cite{Lust:2017bgx}. 
Further aspects of nonassociativity in quantum mechanics along these lines have been recently discussed in~\cite{Bojowald:2014oea,Bojowald:2015cha,Heninger:2018guf,Bojowald:2018qqa,Bojowald:2020whs}.

A rigorous geometric definition of Jackiw's nonassociative magnetic translation operators in terms of weak projective 2-representations of the translation group on a 2-Hilbert space was given in~\cite{Bunk:2018qvk}, and more generally in terms of a 2-group extension of the translation group in~\cite{Bunk:2020rju}. 
The controlled type of nonassociativity that appears in these examples  is not `homotopy associativity', as seen with $A_\infty$-algebras, but rather `weak associativity', as seen with algebras in braided monoidal categories~\cite{Barnes:2014ksa}. 
See~\cite{Szabo:2017yxd,Szabo:2019gvu,Szabo:2019hhg} for recent reviews of these and other aspects of nonassociativity in quantum mechanics, field theory and string theory. 

To handle nonassociative algebras of observables in quantum mechanics, the notion of composition product was introduced in~\cite{Mylonas:2013jha} as an overarching associative algebra structure corresponding to the nonassociative deformation quantization of almost Poisson structures on phase space. In~\cite{Kupriyanov:2018xji} it was shown to arise as the associative deformation quantization of a symplectic embedding of the nonassociative phase space. The composition product in this case extends the nonassociative algebra of functions with a star-product to an associative algebra of differential operators; see~\cite{Kupriyanov:2018xji} for an explicit description of this algebra in the case of a uniform magnetic charge distribution. The GNS construction for associative deformation quantization  considered in~\cite{Bordemann:1997tz} then provides a natural candidate for a (pre-)Hilbert space in nonassociative quantum mechanics. 

In this paper we develop a more precise and rigorous description of the general notion of composition product in terms of enveloping algebras of nonassociative algebras. We use it to formulate properties of states together with notions of trace in a novel and interesting way. We are even able to develop GNS constructions within our formalism, contrary to the claim~\cite{Bojowald:2014oea} that this is not possible in nonassociative quantum mechanics.

Finite-dimensional nonassociative algebras are also abundant in mathematics. The best known example arises in the classification of the normed division algebras over the real numbers. There are only four such algebras, occuring in dimensions one, two, four and eight, represented by the fields $\FR$ and $\FC$ of real and complex numbers, respectively, the noncommutative ring of quaternions $\mathbbm{H}$, and the nonassociative algebra of octonions $\mathbbm{O}$. The latter comes with a weakened form of nonassociativity: $\mathbbm{O}$ is an \emph{alternative} algebra, i.e. it is two-element associative. The associator of an alternative algebra is an alternating trilinear form, and any alternative algebra is also \emph{power associative}, i.e. it is one-element associative. Another well-known class of  algebras are the Jordan algebras, which enjoy the similar property of being \emph{flexible}. 

Both classes of algebras have a long history in attempts to develop a proper theory of nonassociative quantum mechanics. 
Let us stress though that our approach is fundamentally different. While we develop a suitably general formulation of ordinary quantum mechanics that can also handle Jordan and octonion-type observable and state algebras, we do not consider a Jordan-type reformulation of traditional quantum mechanics, nor do we consider quantum mechanics with octonion-valued amplitudes or other such drastic modifications of quantum theory.

In this paper we take an algebraic approach to nonassociative quantum mechanics, which has the advantage of avoiding various analytic issues as well as the need for delving into the technicalities of higher geometry and algebra. It allows us to focus directly on the challenges posed by nonassociativity without analytical distractions that play a more technical role in our context; such issues are of course important and should ultimately be addressed. The main examples we consider will in fact deal with finite-dimensional algebras, where these issues are absent from the outset. Our approach is however general enough that it further avoids the need of additional simplifying properties of nonassociative algebras, such as alternativity, which are usually introduced as weakened forms of nonassociativity in order to simplify algebraic manipulations (as in e.g.~\cite{Bojowald:2014oea,Bojowald:2015cha}), but do not feature in the infinite-dimensional physical examples discussed above.

\newpage

\subsection*{\underline{\sf Outline and summary of results}}

The remainder of this paper is organised into seven sections, structured as follows:

\begin{myitemize}

\item In \S\ref{sec:enveloping} we introduce one of the main algebraic players of this paper: the universal enveloping algebra $\CCU(\CA)$ of a nonassociative algebra $\CA$, and its composition product. One point of the universal enveloping algebra is that it can be considered as a starting point for general birepresentations of $\CA$. In this paper we mostly ignore this level of generality and work with the regular birepresentation, which defines the associative multiplication algebra of $\CA$. In this setting alone, the underlying observer algebra can be quite general, and some of our considerations can be applied as a further generalization of our approach to quantum mechanics. In this paper we specialise to the case of unital $\ast$-algebras $\CA$, and all our examples are of this type. This naturally identifies the module actions of $\CCU(\CA)$ by left and right multiplication, hence simplifying some of the treatment, and also making closer contact to more traditional approaches to quantum mechanics.

\item In \S\ref{sec:GNS} we apply some standard constructions from the algebraic formulation of quantum mechanics to the associative unital $\ast$-algebra $\CCU(\CA)$. We define the notion of states over $\CCU(\CA)$, and the GNS representation of $\CCU(\CA)$ on the corresponding pre-Hilbert space. The latter can be used to distinguish between pure and mixed states, which may also be characterised through an algebraic notion of entropy.

\item In \S\ref{sec:eigen} we introduce the notion of observables of the enveloping algebra $\CCU(\CA)$, and derive uncertainty relations between them. We give an algebraic definition of eigenstate and reproduce several standard properties from basic quantum mechanics.

\item To zoom in on truly nonassociative features implied by the original observer algebra $\CA$, in \S\ref{sec:trace} we assume that $\CA$ is equipped with a suitable trace functional. Like the states, a trace captures a natural notion of positivity essential for a probabilistic interpretation of the theory. We use traces on $\CA$ to define states on $\CCU(\CA)$, which when restricted to $\CA$ provide a density matrix description of nonassociative quantum mechanics. These connect the notion of eigenstates of elements of $\CCU(\CA)$ with a more familiar notion of eigenvectors in $\CA$. The tracial state is also used to provide an alternative GNS representation of $\CCU(\CA)$, which employs a smaller pre-Hilbert space constructed directly from the nonassociative algebra $\CA$.

\item In \S\ref{sec:CPM} we formulate completely positive dynamics for states over the enveloping algebra $\CCU(\CA)$. In the case of unitary time evolution, we derive analogues of the Heisenberg and Schr\"odinger equations, as well as the Lindblad master equation, in nonassociative quantum mechanics. 

\item In \S\ref{sec:Jordan} we apply our general considerations to the explicit example of a finite-dimensional matrix algebra endowed with the Jordan product. We give an explicit description of its enveloping algebra, its tracial GNS representation and its eigenstates. We provide an example of a genuinely nonassociative observable (which would vanish in traditional quantum mechanics), and use it as a Hamiltonian to obtain solutions of the Schr\"odinger and Heisenberg equations in this model.

\item The algebra of octonions is treated in \S\ref{sec:oct} as another explicit example. Again we give an explicit description of its enveloping algebra and its tracial GNS representation. We also display an explicit pair of observables for which the tracial state is a minimum uncertainty state.

\item Finally, we conclude in \S\ref{sec:conclusion} with a discussion of some of the physical implications of our scheme.

\end{myitemize}

\subsection*{\underline{\sf Notation}}

\begin{myitemize}

\item All tensor products are taken over a ground field $\fk$ of characteristic zero. All homomorphisms are $\fk$-linear.

\item Generic elements of an algebra $\CA$ are denoted $a,b,c,\dots$; we write $\widehat a,\widehat b,\widehat c,\dots$ for their images in the corresponding enveloping algebra $\CCU(\CA)$ under the injection $\CA \longhookrightarrow\CCU(\CA)$ of vector spaces. Elements of the opposite algebra $\CA^{\rm op}$ are denoted $a',b',c',\dots$. Generic elements of $\CCU(\CA)$ are denoted $A,B,C,\dots$. Algebraic objects associated to an algebra $\CA$ are denoted by caligraphic symbols $\CH,\CI,\dots$, those associated to the enveloping algebra $\CCU(\CA)$ by script symbols $\CCH,\CCI,\dots$.

\item We use the superscript ${}^*$ to denote the $\ast$-involution on a $\ast$-algebra, while ${}^\dag$ refers to the adjoint of an operator on a pre-Hilbert space. An overline \smash{$\overline{ \phantom{A} }$} denotes complex conjugation on the field $\FC$ of complex numbers. 

\item We denote by $\unit$ the unit element of a unital algebra, and by $\id$ the identity endomorphism of a vector space.

\end{myitemize}

%

\section{Enveloping algebras and composition products}
\label{sec:enveloping}

The basic example of a nonassociative algebra is a Lie algebra $\frg$ over a ground field $\fk$ of characteristic zero, where the violation of associativity $[[a,b],c]\neq [a,[b,c]]$ for $a,b,c\in\frg$ is controlled by the Jacobi identity for the Lie bracket $[\,\cdot\,,\,\cdot\,]$ of $\frg$. To any Lie algebra $\frg$ there is a ``smallest'' associative unital $\fk$-algebra $\CCU(\frg)$ containing $\frg$, called the universal enveloping algebra of $\frg$. It is constructed as the quotient of the free unital algebra 
on the underlying $\fk$-vector space of $\frg$, usually identified with the tensor algebra $
\CCT(\frg) = \fk \, \oplus \ \bigoplus_{n\geqslant1}\,\frg^{\otimes n}
$,
by the two-sided ideal generated by elements of the form $a\otimes b-b\otimes a -[a,b]$ for $a,b\in\frg$. This construction uses the property that any associative algebra has a Lie algebra structure on its underlying vector space given by the commutator bracket. We wish to formalise this notion to more general nonassociative algebras. 

Enveloping algebras are naturally related to modules; for example, the representations of a Lie algebra $\frg$ are in one-to-one correspondence with representations of its universal enveloping algebra $\CCU(\frg)$.  One approach to modules over a general nonassociative algebra $\CA$ is to study representations of its enveloping algebra, formulated in the language of operads~\cite{Ginzburg1994} (see e.g.~\cite{Osborn78,KoSh18,Dhankar23} for other approaches to representations of nonassociative algebras). For this, one works in a category of algebras over some operad $\CCP$. For an algebra $\CA$ and a vector space $V$, an $\CA$-module structure on $V$ is given as a collection of operations defined using all possible operations from $\CCP$, by inserting an element of $V$ into one slot and elements of $\CA$ into all other slots. The module axioms follow from taking the defining identities of $\CCP$ and forming new identities, marking one element in all possible ways, while treating unmarked elements as belonging to $\CA$ and marked elements as belonging to $V$. This construction provides a good notion  of `enveloping algebra' $\CCU(\CA)$ for the nonassociative algebra $\CA$.

When ${\CCP}$ is the commutative operad, this recovers the usual notion of a module over a commutative algebra $\CA$, whereas when ${\CCP}$ is the associative operad it recovers the notion of an $\CA$-bimodule. When $\CCP$ is the Lie operad, it coincides with the usual notion of representation of a Lie algebra; the natural map from the Lie operad to the associative operad yields the forgetful functor from associative algebras to Lie algebras, whose left adjoint is the universal enveloping algebra functor. Whereas a Lie algebra structure, as well as a commutative Jordan algebra structure, can always be extracted from an associative algebra structure, we do not know of any systematic way of doing so for a general nonassociative algebra structure. Instead, we will adopt a different definition of universal enveloping algebra that is tailored to more general cases, and which agrees with the usual notion in the case of Lie algebras.

\subsection*{\underline{\sf The universal enveloping algebra}}

Let $(\CA,\unit)$ be a unital algebra over a field $\fk$ of characteristic zero, with multiplication denoted by juxtaposition. This means that $\CA$ is a $\fk$-vector space and the multiplication $\CA\otimes\CA\longrightarrow\CA$ is a linear map, which is neither necessarily commutative nor associative, i.e. in general $a\,b\neq b\,a$ and $a\, (b\, c) \neq (a\, b)\, c$, for $a,b,c\in \CA$. The unit element $\unit\in\CA$ is the unique element with the property $\unit\,a=a=a\,\unit$ for all $a\in\CA$. We will often simply write $\CA$ for the pair $(\CA,\unit)$ as shorthand. Let $\CA^{\rm op}$ be the opposite algebra of $\CA$; it is anti-isomorphic to $\CA$ under an involutive isomorphism $a\longmapsto a'$ of underlying vector spaces with multiplication $a'\,b' = (b\,a)'$. 

Let 
\begin{align}
\CCU(\CA) = \CCF(\CA\oplus\CA^{\rm op})
\end{align}
be the free associative algebra over $\fk$ on the underlying vector space of $\CA\oplus\CA^{\rm op}$. It inherits the unit $\unit$ of $\CA$ which makes it into an associative unital algebra.  The left modules $V$ over the associative algebra $\CCU(\CA)$ define \emph{birepresentations} of the nonassociative algebra $\CA$, see~\cite{Jacobson1954}. By construction there is a map
\begin{align} \label{eq:injectivemap}
\widehat{ \, \ \cdot \ \, }\,:\CA\oplus\CA^{\rm op}\longrightarrow\CCU(\CA) \ , \quad (a,b')\longmapsto\widehat a + \widehat{b'}
\end{align}
of underlying vector spaces. 

In this paper we shall only work with the regular birepresentation $V=\CA$ defined by left and right multiplication of $\CA$ on itself. 
Thus we regard $\CA$ as a left $\CCU(\CA)$-module under the left action $\triangleright:\CCU(\CA)\times \CA\longrightarrow \CA$ defined by left multiplication
\begin{align}\label{eq:leftmodule}
\widehat a\,\triangleright\, x :=  a\, x \qquad \mbox{and} \qquad \widehat{a'}\,\triangleright\, x :=  x\, a  \ ,
\end{align}
on generators $a\in \CA$, for all $x\in \CA$; note that the left action by $\CA^{\rm op}$ is equivalent to the right action of $\CA$ by right multiplication. Because $\CA$ is unital, the map $a\longmapsto\widehat a$ of $\CA$ into $\CCU(\CA)$ is injective and we may identify $\CA$ with its image in $\CCU(\CA)$ under this map. Similarly, \smash{$a'\longmapsto\widehat{a'}$} is injective and we may identify $\CA^{\rm op}$ with its image in $\CCU(\CA)$. 

The idea behind this definition is that, for any algebra $\CA$, there is always the natural associative algebra $\sEnd(\CA)$ of endomorphisms of the underlying vector space of $\CA$ with the composition of $\fk$-linear maps. The left module actions define $\fk$-linear maps $\CA\longrightarrow\sEnd(\CA)$ through $a\longmapsto \widehat a\, \triangleright\, (\,\cdot\,)$ and \smash{$a\longmapsto \widehat{a'}\, \triangleright\, (\,\cdot\,)$} for $a\in\CA$. Viewed in this way as a subalgebra of $\sEnd(\CA)$, the associative algebra $\CCU(\CA)$ defines the \emph{multiplication algebra} of $\CA$, see~\cite{Schafer55,Schafer95}.  If $\CA$ has finite dimension $n$, then the multiplication algebra is isomorphic to a subalgebra of the algebra $ \mathbbm{M}_n(\fk)$ of $n{\times}n$ matrices with entries in $\fk$. In the following we shall sometimes abuse notation (and terminology) by identifying $\CCU(\CA)$ with the multiplication algebra. 

\begin{definition} \label{def:compalg}
The associative unital $\fk$-algebra $(\CCU(\CA),\circ,\unit)$ is the \emph{enveloping algebra} of $(\CA,\unit)$, whose multiplication $\circ:\CCU(\CA)\otimes\CCU(\CA)\longrightarrow\CCU(\CA)$ is the \emph{composition product}.
\end{definition}

By construction, the associative algebra $\CCU(\CA)$ consists of linear combinations of elements of the form \smash{$(\widehat a_1+\widehat b'_1)\circ (\widehat a_2+\widehat b'_2)\circ \cdots \circ (\widehat a_n+\widehat b'_n)$} with $a_i,b_i\in \CA$, which by definition act on $x\in\CA$ as
\begin{align}
(\widehat a_1\circ \widehat a_2\circ \cdots \circ \widehat a_n)\,\triangleright\, x = \widehat a_1\,\triangleright\,\big(\widehat a_2\,\triangleright\,\big(\cdots\,\triangleright(\widehat a_n\,\triangleright\,x)\cdots\big)\big) = a_1\,\big( a_2\,\big(\cdots (a_n\,x)\cdots\big)\big) \ ,
\end{align}
and 
\begin{align}
(\,\widehat b'_1\circ \widehat b'_2\circ \cdots \circ \widehat b'_n)\,\triangleright\, x = \widehat b'_1\,\triangleright\,\big(\,\widehat b'_2\,\triangleright\,\big(\cdots\,\triangleright(\,\widehat b'_n\,\triangleright\,x)\cdots\big)\big) = 
\big(\big(\cdots(x\,b_n)\cdots\big)\,b_2\big)\, b_1 \ .
\end{align}

The involutive anti-isomorphism $a\longmapsto a'$ between $\CA$ and $\CA^{\rm op}$ extends as an anti-automorphism $A\longmapsto A'$ of $\CCU(\CA)$ defined by
\begin{align}
(\widehat a_1\circ \widehat a_2\circ \cdots \circ \widehat a_n)' := \widehat a_n'\circ\cdots\circ\widehat a_2'\circ\widehat a_1' \ .
\end{align}
This involution interchanges the left module actions by left and right multiplication, so that
\begin{align}
(\widehat a_1\circ \widehat a_2\circ \cdots \circ \widehat a_n)'\,\triangleright\, x = \big(\big(\cdots(x\,a_1)\,a_2\big)\cdots\big)\, a_n \ ,
\end{align}
together with the obvious extensions to mixed composition products among elements from the embeddings of $\CA$ and $\CA^{\rm op}$; for example
\begin{align}
(\widehat a_i\circ\widehat b'_j)'\,\triangleright\, x = (\,\widehat b_j\circ\widehat a_i')\,\triangleright\, x = b_j\,(x\,a_i) \ ,
\end{align}
and so on by iteration of these actions.

The injective linear map \eqref{eq:injectivemap} extends to an epimorphism of associative unital $\fk$-algebras
\begin{align} \label{eq:epimor}
{\tt P}:\CCT_0(\CA\oplus\CA^{\rm op}) \relbar\joinrel\twoheadrightarrow \CCU(\CA)
\end{align}
from the reduced tensor algebra $\CCT_0(\CA\oplus\CA^{\rm op}) = \bigoplus_{n\geqslant1}\,(\CA\oplus\CA^{\rm op})^{\otimes n}$ of the underlying vector space of $\CA\oplus\CA^{\rm op}$, which comes with a canonical inclusion $\CA\oplus\CA^{\rm op}\lhook\joinrel\longrightarrow\CCT_0(\CA\oplus\CA^{\rm op})$. This identifies $\CCU(\CA)\simeq\CCT_0(\CA\oplus\CA^{\rm op})/\ker({\tt P})$, where in general the kernel of $\tt P$ depends on the relations in $\CA$. 

We can show that $\CCU(\CA)$ is the most general such algebra containing $\CA\oplus\CA^{\rm op}$ through

\begin{proposition} \label{prop:universal}
The enveloping algebra $\CCU(\CA)$ is \emph{universal}: any $\fk$-linear injective map $\iota:\CA\oplus\CA^{\rm op}\lhook\joinrel\longrightarrow\CCA$ of $\CA\oplus\CA^{\rm op}$ into another associative unital $\fk$-algebra $\CCA$ factors through a unique algebra homomorphism ${\tt T}:\CCU(\CA)\longrightarrow\CCA$ according to the commutative diagram
\begin{align}
\begin{split}
\xymatrix{
\CCU(\CA) \ar[rr]^{\tt T} & & \CCA \\
& \CA\oplus\CA^{{\rm op}^{\phantom{\dag}}}  \ar@{^{(}->}[ul]^{\widehat{  \ \cdot \ }} \ar@{_{(}->}[ur]_\iota &
}
\end{split}
\end{align}
\end{proposition}

\begin{proof}
Since $\CCU(\CA)$ is generated by left and right multiplication on the algebra $\CA$ through the left module actions \eqref{eq:leftmodule}, it suffices to define $\tt T$ on the image of the injection \eqref{eq:injectivemap} by ${\tt T}(\widehat a+\widehat{b'}) = \iota(a,b')$ for all $a,b\in\CA$.  Since $\widehat a+\widehat{b'}$ generate $\CCU(\CA)$, we can then extend $\tt T$ uniquely to all of $\CCU(\CA)$ as an algebra homomorphism on arbitrary composition products.
\end{proof}

\subsection*{\underline{\sf Examples}}

Let us look at how this construction works in the three standard examples discussed at the beginning of this section.

\begin{example}[\textbf{Lie algebras}] \label{ex:Liecom}
Let $\frg$ be a Lie algebra over $\fk$ and denote its Lie bracket by $[\,\cdot\,,\,\cdot\,]$. The unitization
\begin{align}
\CA = \fk\oplus\frg
\end{align}
over $\fk$ is a unital $\fk$-algebra with multiplication defined by
\begin{align}
(k,a)\,(l,b) = (k\,l\,,\,l\,a + k\,b + [a,b])
\end{align}
and unit element $\unit=(1,0)$. There is a natural algebra monomorphism $i:\frg\lhook\joinrel\longrightarrow\CA$ given by $a\longmapsto(0,a)$, and in general we write $(k,a)$ as $k+a$.

By skew-symmetry of the Lie bracket, it follows that the left module structure of $\CCU(\CA)$ restricted to $\frg$ satisfies
\begin{align}
\widehat a\,\triangleright\,x = [a,x] = -[x,a] = -\widehat{a'}\,\triangleright\,x \ .
\end{align}
Hence the embeddings of $\CA$ and $\CA^{\rm op}$ differ only through a sign and can essentially be identified with each other:
\begin{align} \label{eq:equalsign}
\widehat a = -\widehat{a'} \ .
\end{align}

On the other hand, the Jacobi identity for the Lie bracket implies
\begin{align}
\big(\widehat a \circ\widehat b-\widehat b\circ\widehat a\big)\,\triangleright\,x = [a,[b,x]] - [b,[a,x]] = [[a,b],x] = \widehat{[a,b]}\,\triangleright\,x \ ,
\end{align}
for $a,b,x\in\CA$.
This shows that the enveloping algebra $\CCU(\CA)$ represents the Lie algebra $\frg$ through
\begin{align} \label{eq:Liealgrep}
\widehat{[a,b]} = \widehat a \circ\widehat b-\widehat b\circ\widehat a \ .
\end{align}
Through the epimorphism \eqref{eq:epimor}, the quotient of the reduced tensor algebra $\CCT_0(\CA\oplus\CA^{\rm op})$ by the equations \eqref{eq:equalsign} and \eqref{eq:Liealgrep} defines the usual Lie universal enveloping algebra. 

As a subalgebra of $\sEnd(\CA)$, the enveloping algebra $\CCU(\CA)$ defines the unitization of the algebra ${\sf ad}(\frg)$ of inner derivations of the Lie algebra $\frg$, i.e.~the adjoint representation of $\frg$. 
\end{example}

\begin{example}[\textbf{Associative algebras}] \label{lem:bimodcomp}
Let $\CA$ be an associative unital $\fk$-algebra with enveloping algebra $\CCU(\CA)$. Then
\begin{align}  \label{eq:nobimod}
\widehat{b'}\,\triangleright\,(\widehat a\,\triangleright\, x) = (a\,x)\,b =a\,(x\,b) = \widehat a\,\triangleright\,(\widehat{b'}\,\triangleright\, x) \ ,
\end{align}
hence the embeddings of $\CA$ and $\CA^{\rm op}$ commute with each other:
\begin{align} \label{eq:asscommute}
\widehat{b'}\circ\widehat a = \widehat a\circ\widehat{b'} \ .
\end{align}
On the other hand, if $\CA$ were nonassociative then the second equality of \eqref{eq:nobimod} would fail to hold in general. It follows that associativity of $\CA$ is equivalent to the left actions by $\CCU(\CA)$ combining to a bimodule structure on $\CA$ inherited from the natural $\CA$-bimodule structure induced by multiplication of $\CA$ on itself.

The simple computation
\begin{align} \label{eq:comprep}
(\widehat a\circ \widehat b\,)\,\triangleright\,x = a\,(b\,x) = (a\,b)\,x = \widehat{a\,b}\,\triangleright\, x 
\end{align}
shows that
\begin{align} \label{eq:algrep}
\widehat a\circ \widehat b = \widehat{a\,b} \ ,
\end{align}
for all $a,b,x\in\CA$. On the other hand, if $\CA$ were nonassociative then the second equality of \eqref{eq:comprep} would fail to hold in general. It follows that associativity of $\CA$ is equivalently characterised by its enveloping algebra $\CCU(\CA)$ defining a faithful algebra representation of $\CA$ in $\sEnd(\CA)$. Similarly
\begin{align} \label{eq:algoprep}
\widehat{(a\,b)'} = \widehat{b'}\circ\widehat{a'} \ .
\end{align}

The quotient of $\CCT_0(\CA\oplus\CA^{\rm op})$ by the relations \eqref{eq:asscommute}, \eqref{eq:algrep} and \eqref{eq:algoprep} identifies the enveloping algebra of $\CA$ with the extended algebra
\begin{align}
\CCU(\CA)\simeq\CA\otimes\CA^{\rm op} \ ,
\end{align}
as expected from the universal property (Proposition~\ref{prop:universal}). In this representation $\widehat a = a\otimes\unit$ and \smash{$\widehat{a'} = \unit\otimes a$}.
\end{example}

\begin{example}[\textbf{Commutative algebras}] \label{lem:bimodcomp2}
Let $\CA$ be a commutative (but not necessarily associative) unital $\fk$-algebra with enveloping algebra $\CCU(\CA)$. Then 
\begin{align}
\widehat a\,\triangleright\,x = a\,x = x\,a = \widehat{a'}\,\triangleright\, x \ .
\end{align}
On the other hand, if $\CA$ were noncommutative then the second equality would fail to hold in general. It follows that commutativity of $\CA$ is equivalent to equality of the embeddings of $\CA$ and $\CA^{\rm op}=\CA$:
\begin{align}
\widehat a = \widehat{a'} \ ,
\end{align}
and hence $\CCU(\CA)\simeq\CCT_0(\CA)$. Note that in general $\CCU(\CA)$ is noncommutative, unless $\CA$ is associative (see Example~\ref{lem:bimodcomp}) in which case $\CCU(\CA)\simeq\CA$.
\end{example}

\subsection*{\underline{\sf Enveloping $\ast$-algebras}}

Similarly to Examples~\ref{ex:Liecom},~\ref{lem:bimodcomp} and~\ref{lem:bimodcomp2}, in the general case we would also like to have a means of naturally identifying the embeddings of $\CA$ and $\CA^{\rm op}$ in $\CCU(\CA)$, as well as their left actions on $\CA$. We henceforth set $\fk=\FC$ to be the ground field of complex numbers, and always assume that our algebras are equipped with $\ast$-structures. 
The importance of this choice in the traditional approach to quantum mechanics stems from the realisation that real-valued wavefunctions  only describe states with zero momentum, or zero magnetic quantum number in the case of spin. The necessity of complex-valued amplitudes can be traced back to the fact that the Schr\"odinger equation is a first-order equation. The real and imaginary part (or modulus and phase) of a wavefunction play a role analogous to the need for prescibing position as well as momentum to describe a classical state of a particle. Attempts to use other ground fields (such as the real numbers $\FR$), or quaternionic and even octonionic division algebras, for quantum mechanics can be found in the literature and will not be pursued here.

Let $(\CA,{}^\ast,\unit)$ be a unital $\ast$-algebra over $\FC$, where the $\ast$-involution ${}^\ast:\CA\longrightarrow\CA$ is an antilinear antihomomorphism,
\begin{align}
(a^*)^*=a \ \text{,} \quad (z\,\unit)^*=\overline{z}\,\unit \qquad \text{and} \qquad (a \, b)^\ast = b^\ast\, a^\ast \ ,
\end{align}
for all $a,b\in\CA$ and $z\in\FC$.
Again we will often simply write $\CA$ for the triple $(\CA,{}^\ast,\unit)$ when the operations are clear from the context. Via the same constructions as before, the $\ast$-involution on $\CA$ induces a $\ast$-involution on its enveloping algebra $\CCU(\CA)$, which for simplicity we denote with the same symbol. We call the associative unital $\ast$-algebra $(\CCU(\CA),\circ,{}^*,\unit)$ the \emph{universal enveloping $\ast$-algebra} of $(\CA,{}^\ast,\unit)$; for brevity, we shall often simply call it the `enveloping algebra'.

The role of the $\ast$-involutions is to generally identify the left actions of $\CCU(\CA)$ on $\CA$ by opposite algebra structures through

\begin{lemma} \label{lem:leftrightmod}
The left module structures on $\CA$ are related as
\begin{align}
(A'\,\triangleright\, x)^* = A^*\,\triangleright\, x^* \ ,
\end{align}
for all $A\in\CCU(\CA)$ and $x\in \CA$.
\end{lemma}

\begin{proof}
This follows easily by iterating the basic relations
\begin{align}
\big((\widehat a\circ \widehat b\,)\,\triangleright\, x\big)^\ast = \big(a\,(b\, x)\big)^\ast = (x^\ast\,b^\ast)\,a^\ast = \big(\,\widehat{a'}{}^\ast\circ \widehat{b'}{}^\ast\big)\,\triangleright\,x^\ast = \big(\widehat{b'}\circ \widehat{a'}\,\big)^\ast\,\triangleright\,x^\ast \ ,
\end{align}
and similarly
\begin{align}
\big((\widehat{a'}\circ\widehat{b}\,)\,\triangleright\,x\big)^* = \big((b\,x)\,a\big)^* = a^*\,(x^*\,b^*) = \big(\widehat a^* \circ \widehat{b'}{}^*\big)\,\triangleright\,x^* = \big(\widehat{b'}\circ \widehat{a}\,\big)^\ast\,\triangleright\,x^\ast \ ,
\end{align}
for $a,b,x\in\CA$. 
\end{proof}

Lemma~\ref{lem:leftrightmod} can be used to identify the embeddings of $\CA$ and $\CA^{\rm op}$ into $\CCU(\CA)$, so that in the following it will suffice for the most part to work with the injection \smash{$\widehat{\CA}\subset\CCU(\CA)$} alone.


\section{States and GNS constructions}
\label{sec:GNS}

Let $(\CA,{}^*,\unit)$ be a unital $\ast$-algebra over $\FC$. Given the associative unital $\ast$-algebra $(\CCU(\CA),\circ,{}^\ast,\unit)$, we can now carry out some standard constructions on it. We follow the approach of~\cite{Khavkine:2014mta} which develops the Gel'fand-Naimark-Segal (GNS) construction for $\ast$-algebras, generalized from the usual setting of $C^*$-algebras. While this uses very little of the underlying nonassociative structure of our original algebra $\CA$, it is always available, and it can handle the nonassociative operators and observables discussed in \S\ref{sec:Intro}. We shall discuss how to treat the \emph{bona fide} nonassociative quantum mechanics later on in \S\ref{sec:trace}.

Regarding $\CCU(\CA)$ as an observer algebra, a state can be characterised by its measurement outcomes.

\begin{definition} \label{def:state}
A \emph{state} $\omega$ over $\CCU(\CA)$ is a $\FC$-linear functional $\omega:\CCU(\CA)\longrightarrow\FC$ which is positive, i.e. $\omega(A^\ast\circ A)\geqslant0$ for all $A\in\CCU(\CA)$, and normalized, i.e. $\omega(\unit)=1$. The space of all states over $\CCU(\CA)$ is denoted $\mathbb{E}(\CCU(\CA))$.
\end{definition}

\begin{remark} \label{rem:pos1}
This definition of state gives us our first notion of positivity in nonassociative quantum mechanics. 
\end{remark}

\begin{lemma}[\textbf{Cauchy-Schwartz inequality}]\label{lem:CSinequality}
The map $(A,B)\longmapsto\omega(A^\ast\circ B)$ is a semi-definite sesquilinear form for any state $\omega$, which satisfies
\begin{align}\label{eqlem:conj}
\omega(A^\ast\circ B) = \overline{\omega(B^\ast\circ A)}
\end{align}
as well as
\begin{align}\label{eqlem:CSineq}
\big|\omega\big(A^\ast\circ B\big)\big|^2 \leqslant \omega\big(A^\ast\circ A\big)\,\omega\big(B^\ast\circ B\big) \ ,
\end{align}
for all $A,B\in\CCU(\CA)$. 
\end{lemma}

\begin{proof}
The first statement follows from $\FC$-linearity of the state $\omega$. For $A,B\in\CCU(\CA)$ and a complex vector $\underline{z} = (z_1,z_2) \in\FC^2$, positivity of $\omega$ implies that
\begin{align}
0\leqslant\omega\big((z_1\,A+z_2\,B)^\ast\circ(z_1\,A+z_2\,B)\big) = \underline{z}^\dag \ \underline{\tt Q} \ \underline{z}
\end{align}
defines a positive semi-definite quadratic form on $\FC^2$ with matrix
\begin{align}
\underline{\tt Q} = \begin{pmatrix}
\omega(A^\ast\circ A) & \omega(A^\ast\circ B) \\ \omega(B^\ast\circ A) & \omega(B^\ast\circ B)
\end{pmatrix} \ .
\end{align}
Hence $\underline{\tt Q}$ must be a positive semi-definite matrix. From symmetry $\underline{\tt Q}^\dag = \underline{\tt Q}$ we find \eqref{eqlem:conj}, and from the determinant
\begin{align}
0\leqslant\det\underline{\tt Q} = \omega(A^\ast\circ A) \, \omega(B^\ast\circ B) - \omega(A^\ast\circ B) \, \omega(B^\ast\circ A)
\end{align}
we find \eqref{eqlem:CSineq}.
\end{proof}

It follows from Lemma~\ref{lem:CSinequality} that the subspace
\begin{align}
\CCI_\omega := \big\{ A\in\CCU(\CA) \ \big| \ \omega\big(A^\ast\circ A\big)=0\big\}
\end{align}
is a left ideal of $\CCU(\CA)$, called the \emph{Gel'fand ideal} of $\omega$. We can then form the quotient space
\begin{align}
\CCH_\omega := \CCU(\CA) \, \big\slash \, \CCI_\omega \ ,
\end{align}
and denote the equivalence classes in $\CCH_\omega$ by
\begin{align}
\psi_A := \big\{\widetilde A\in\CCU(\CA) \  \big| \ \widetilde A - A\in\CCI_\omega \big\} \ .
\end{align}
On $\CCH_\omega$ we define an inner product 
\begin{align}
(\psi_A , \psi_B)_\omega := \omega\big(A^\ast\circ B\big) \ ,
\end{align}
which is well-defined because $\CCI_\omega$ is a left ideal, and makes $\CCH_\omega$ into a pre-Hilbert space over $\FC$. 

We now define a linear {representation} $\pi_\omega:\CCU(\CA)\longrightarrow\sEnd(\CCH_\omega)$ by
\begin{align}
\pi_\omega(A)\,\psi_B := \psi_{A\circ B} \ .
\end{align}
This is well-defined, since $\CCI_\omega$ is a left ideal and hence $\CCH_\omega$ is a left $\CCU(\CA)$-module, and gives a $\ast$-representation: $\pi_\omega(A\circ B) = \pi_\omega(A) \, \pi_\omega(B)$ and $\pi_\omega(A^\ast) =  \pi_\omega(A)^\dag$, by associativity of $\CCU(\CA)$. Furthermore, the representation $\pi_\omega$ is {cyclic} with vacuum vector $\psi_\unit$: Every vector $\psi_A\in\CCH_\omega$ can be written as $\psi_A=\psi_{A\circ\unit}=\pi_\omega(A)\,\psi_\unit$, and
\begin{align}
\langle A\rangle_\omega := \omega(A) = (\psi_\unit , \pi_\omega(A)\,\psi_\unit)_\omega \ .
\end{align}

\begin{definition}
The cyclic $\ast$-representation $(\pi_\omega,\psi_\unit)$ is the \emph{GNS representation} of the enveloping algebra $(\CCU(\CA),\circ,{}^\ast,\unit)$ on the pre-Hilbert space $\CCH_\omega$. 
\end{definition}

\begin{remark}
The GNS representation is unique up to unitary equivalence. In the case that $\CCU(\CA)$ is a unital $C^*$-algebra, $\CCH_\omega$ is a Hilbert space over $\FC$, i.e. every Cauchy sequence in $\CCH_\omega$ converges in $\CCH_\omega$ with respect to the norm induced by the inner product $(\,\cdot\,,\,\cdot\,)_\omega$, while the representation $\pi_\omega$ takes values in the bounded operators on $\CCH_\omega$ and is norm-decreasing (hence continuous) with respect to the operator norm.
\end{remark}

\begin{remark} \label{rem:unitvector}
Let $\omega\in\mathbb{E}(\CCU(\CA))$ be a reference state, and $\psi\in\CCH_\omega$ a fixed unit vector, $(\psi,\psi)_\omega=1$. Then the assignment
\begin{align}
A\longmapsto\omega_\psi(A) := \big(\psi,\pi_\omega(A)\,\psi\big)_\omega
\end{align}
defines a state $\omega_\psi\in\mathbb{E}(\CCU(\CA))$.
\end{remark}

\begin{remark} \label{rem:normal}
In the case that $\CCU(\CA)$ is a unital $C^*$-algebra, the GNS construction allows for the definition of a \emph{normal state}: Let $\omega\in\mathbb{E}(\CCU(\CA))$ be a reference state, and let $\rho$ be a positive symmetric trace-class operator on $\CCH_\omega$ which is normalised as $\Tr_{\CCH_\omega}(\rho)=1$. We can then define a new state $\omega_\rho\in\mathbb{E}(\CCU(\CA))$ by 
\begin{align}
\omega_\rho(A) = \Tr_{\CCH_\omega}\big(\rho\,\pi_\omega(A)\big) \ ,
\end{align}
for $A\in\CCU(\CA)$. The operator $\rho$ plays the role of a density matrix in quantum mechanics.
\end{remark}

\subsection*{\underline{\sf Pure and mixed states}}

For any two states $\omega_1$ and $\omega_2$, the convex linear combination $t\,\omega_1+(1-t)\,\omega_2$ with $t\in[0,1]$ is again a state. The space of states $\mathbb{E}(\CCU(\CA))$ is thus a convex body, whose extremal points are called \emph{pure states} and denoted $\mathbb{P}(\CCU(\CA))$; states in the complement $\mathbb{E}(\CCU(\CA))\setminus \mathbb{P}(\CCU(\CA))$ are said to be \emph{mixed}. Any state $\omega\in\mathbb{E}(\CCU(\CA))$ can be decomposed into a convex linear combination of pure states. 

The GNS construction provides a characterisation of pure states through~\cite{Khavkine:2014mta}

\begin{theorem}
A state $\omega\in\mathbb{E}(\CCU(\CA))$ is pure if and only if its GNS representation $(\pi_\omega,\psi_\unit)$ is weakly irreducible.
\end{theorem}

In the case that $\CCU(\CA)$ is a unital $C^*$-algebra, one can drop the adjective `weakly'. If $\omega$ is pure, then the states of Remark~\ref{rem:unitvector} are also pure. 

Another way to characterise the mixing in a state in quantum mechanics is through its entropy. To give an algebraic notion of the entropy of a state, we follow~\cite{Facchi:2021taa}. Given a probability vector $\vec p=(p_1,\dots,p_N)$, i.e. $p_l\geqslant0$ and $\sum_{l=1}^N\,p_l=1$, define its \emph{Shannon entropy} by
\begin{align}
\ttH(\vec p\,) = -\sum_{l=1}^N\,p_l\log p_l \ .
\end{align}

\begin{definition}
The \emph{entropy} $\ttS(\omega)\in\FR_{\geqslant0}$ of a state $\omega\in\mathbb{E}(\CCU(\CA))$ is the minimal Shannon entropy among all decompositions of $\omega$ into pure states:
\begin{align}
\ttS(\omega) = \inf\Big\{\ttH(\vec p\,) \ \Big| \ \omega=\mbox{$\sum\limits_{l=1}^N$}\,p_l\,\omega_l \ , \ p_l\geqslant0 \ , \ \mbox{$\sum\limits_{l=1}^N$}\,p_l=1 \ , \ \omega_l\in\mathbb{P}(\CCU(\CA)) \Big\} \ .
\end{align}
\end{definition}

It follows that pure states can be characterised as having zero entropy, while mixed states have positive entropy.

\section{Observables and their eigenstates}
\label{sec:eigen}

To describe measurements in quantum mechanics, we need to introduce the notion of an observable to measure.

\begin{definition}
An \emph{observable} of the enveloping algebra $\CCU(\CA)$ is an element $O\in\CCU(\CA)$ which is real for the $\ast$-structure on $\CCU(\CA)$, i.e. $O^\ast = O$. The set of observables forms a vector space over $\FR$ which we denote by $\obs(\CCU(\CA))$.
\end{definition}

Observables have real-valued states $\omega$: Applying \eqref{eqlem:conj} with $B=\unit$ shows that 
\begin{align}
\omega(A^\ast)=\overline{\omega(A)}
\end{align}
for all $A\in\CCU(\CA)$, hence $\langle O\rangle_\omega = \omega(O)\in\FR$ for all $O\in\obs(\CCU(\CA))$. Thus any state $\omega$ is a linear functional on the real vector space $\obs(\CCU(\CA))$. Observables also map to symmetric operators in the corresponding GNS representations, as $\pi_\omega(O)=\pi_\omega(O^\ast)=\pi_\omega(O)^\dag$ for $\omega\in\mathbb{E}(\CCU(\CA))$ and $O\in\obs(\CCU(\CA))$.

We introduce the notation
\begin{align} \label{eq:uncertaintydef}
\Delta_\omega O := \sqrt{\omega\big((O-\langle O\rangle_\omega\unit)\circ(O-\langle O\rangle_\omega\unit)\big)}
\end{align}
for the \emph{uncertainty} of the measurement of an observable  $O\in\obs(\CCU(\CA))$ in a state $\omega\in\mathbb{E}(\CCU(\CA))$.

\begin{lemma}[\textbf{Uncertainty relations}] \label{lem:uncert}
Let $O_1,O_2\in\obs(\CCU(\CA))$ be any two observables. Then
\begin{align}
\Delta_\omega O_1 \, \Delta_\omega O_2 \geqslant \tfrac12\,\big|\big\langle[O_1,O_2]_\circ\big\rangle_\omega\big| \ ,
\end{align}
where 
\begin{align}
[O_1,O_2]_\circ =O_1\circ O_2 - O_2\circ O_1 \ .
\end{align}
\end{lemma}

\begin{proof}
By Lemma~\ref{lem:CSinequality} we get
\begin{align}
\begin{split}
(\Delta_\omega O_1)^2\,(\Delta_\omega O_2)^2 &\geqslant \big|\omega\big((O_1-\langle O_1\rangle_\omega\unit)\circ(O_2-\langle O_2\rangle_\omega\unit\big)\big|^2 \\[4pt]
&= \tfrac14\,\big|\langle[O_1,O_2]_\circ\rangle_\omega\big|^2 + \tfrac14\,\big|\langle\{O_1-\langle O_1\rangle_\omega\unit,O_2-\langle O_2\rangle_\omega\unit\}_\circ\rangle_\omega\big|^2 \\[4pt]
&\geqslant \tfrac14\,\big|\langle[O_1,O_2]_\circ\rangle_\omega\big|^2
\end{split}
\end{align}
where
\begin{align}
\{O_1,O_2\}_\circ =O_1\circ O_2 + O_2\circ O_1 \ ,
\end{align}
and the result follows.
\end{proof}

We now adapt the notion of eigenstates in general $C^*$-algebras, which is discussed in~\cite{Rinehart21,DeNittis:2023ppc}, to our setting.

\begin{definition}\label{def:C*eigenvalue}
Let $A\in\CCU(\CA)$. A state $\omega:\CCU(\CA)\longrightarrow\FC$ is an \emph{eigenstate} of $A$ with \emph{eigenvalue} $\lambda\in\FC$ if
\begin{align}
\omega(B\circ A)=\lambda\,\omega(B)
\end{align}
for all $B\in\CCU(\CA)$.
\end{definition}

\begin{remark} \label{rem:linindep}
We can compare eigenstates of any element $A\in\CCU(\CA)$ associated to distinct eigenvalues~\cite{Rinehart21,DeNittis:2023ppc}: Any collection $\{\omega_1,\dots,\omega_N\}\subset\mathbb{E}(\CCU(\CA))$ of eigenstates of $A\in\CCU(\CA)$ with distinct eigenvalues $\{\lambda_1,\dots,\lambda_N\}\subset\FC$ is \emph{linearly independent}, i.e. if $\sum_{i=1}^N\,p_i\,\omega_i=0$ with $p_i\in\FC$, then $p_1=\cdots=p_N=0$.
\end{remark}

\begin{proposition} \label{prop:evreal}
Observables of $\CCU(\CA)$ have real eigenvalues.
\end{proposition}

\begin{proof}
This follows easily by setting $B=\unit$ in Definition~\ref{def:C*eigenvalue} which shows that
\begin{align}
\lambda = \omega(O) = \omega(O^*) = \overline{\omega(O)} = \overline\lambda \ ,
\end{align}
for an eigenstate $\omega\in\mathbb{E}(\CCU(\CA))$ of an observable $O\in\obs(\CCU(\CA))$ with eigenvalue $\lambda\in\FC$.
\end{proof}

By direct computation and the Cauchy-Schwartz inequality, one shows that $\omega$ is an eigenstate of $A$ with eigenvalue $\lambda$ if and only if~\cite{Rinehart21,DeNittis:2023ppc}
\begin{align}\label{eq:eveeqn}
\omega\big((A-\lambda\unit)^\ast\circ(A-\lambda\unit)\big) = 0 \ .
\end{align}
Let $(\pi_\omega,\psi_\unit)$ be the GNS representation of the enveloping algebra $\CCU(\CA)$. It follows from \eqref{eq:eveeqn} that $\omega$ is an eigenstate of $A\in\CCU(\CA)$ with eigenvalue $\lambda = \langle A\rangle_\omega$ if and only if 
\begin{align}
\pi_\omega(A)\,\psi_\unit = \lambda\,\psi_\unit \ ,
\end{align}
or equivalently $\psi_A=\psi_{A\circ \unit}=\lambda\,\psi_\unit$.

We can also easily show that eigenstates correspond to states of zero uncertainty through 

\begin{proposition}
If $\omega\in\mathbb{E}(\CCU(\CA))$ is an eigenstate of an observable $O\in\obs(\CCU(\CA))$, then $\Delta_\omega O=0$.
\end{proposition}

\begin{proof}
If $\lambda=\omega(O)=\langle O\rangle_\omega\in\FR$ is the corresponding eigenvalue, then the result follows immediately from \eqref{eq:eveeqn} and the definition of uncertainty from \eqref{eq:uncertaintydef}.
\end{proof}

\section{Probing nonassociativity: States from traces}
\label{sec:trace}

We now come to our main construction of what should be genuinely called nonassociative quantum mechanics. We wish to adapt the constructions of \S\ref{sec:GNS} and \S\ref{sec:eigen} to the nonassociative algebra $\CA$, without recourse to the full associative enveloping algebra $\CCU(\CA)$. To show that states actually exist in this setting, we introduce the notion of trace on our original unital $*$-algebra $(\CA,{}^*,\unit)$. 

\begin{definition}\label{def:trace}
A \emph{trace} on $\CA$ is a $\FC$-linear functional $\tau:\CA\longrightarrow\FC$ which is
\begin{myitemize}
\item Positive: $\tau(a^*\,a)\geqslant0$;
\item Normalized: $\tau(\unit)=1$;
\item $2$-cyclic: $\tau(a\,b) = \tau(b\,a)$; and
\item $3$-cyclic: $\tau\big(a\,(b\,c)\big) = \tau\big(c\,(a\,b)\big)$;
\end{myitemize}
for all $a,b,c\in\CA$.
\end{definition}

\begin{remark} \label{rem:pos2}
This definition of trace gives us our second notion of positivity in nonassociative quantum mechanics. It does not automatically come with the theory: its existence is an assumption that needs to be checked for a given nonassociative observer algebra.
\end{remark}

Using $2$-cyclicity, the $3$-cyclicity condition can be equivalently written as $\tau\big(a\,(b\,c)\big) = \tau\big((a\,b)\,c\big)$. By induction one can furthermore prove that
\begin{align} \label{eq:tauhigher}
\tau\big(a_1\,(a_2\,(\cdots a_n))\big) = \tau\big(((a_1\,a_2)\cdots)\,a_n\big)
\end{align}
for all $a_1,\dots,a_n\in\CA$. However, this does \emph{not} imply that nonassociativity disappears completely under application of $\tau$. The rebracketed expressions that generally yield distinct traces can be characterised as follows~\cite{Mylonas:2013jha}. 

The total number of ways of pairwise bracketing a product of $n\geqslant1$ elements $a_1,\dots,a_n\in\CA$ is the Catalan number \smash{$C_{n-1} = \frac{(2n-2)!}{(n-1)!\,n!}$}. Starting from the trace \eqref{eq:tauhigher}, by $3$-cyclicity it is equal to a number of traces with different bracketings but generally distinct from any other bracketing with $a_1$ unbracketed on the left of the argument of $\tau$. The distinct ways of bracketing an $n$-fold product of elements of $\CA$ in the argument of the trace $\tau$ is thus organised into $C_{n-2}$ classes, one for each different bracketing where $a_1$ is unbracketed on the left.

\begin{example} 
For $n=4$ there are $C_3=5$ different bracketings, two of which are of the form of the left-hand side of \eqref{eq:tauhigher}. Thus there are $C_2=2$ distinct classes of inequivalent bracketings, namely
\begin{align} \label{eq:4bracketing1}
\tau\big(a_1\,(a_2\,(a_3\,a_4))\big) = \tau\big((a_1\,a_2)\,(a_3\,a_4)\big) = \tau\big(((a_1\,a_2)\,a_3)\,a_4\big)
\end{align}
and
\begin{align}
\tau\big(a_1\,((a_2\,a_3)\,a_4)\big) = \tau\big((a_1\,(a_2\,a_3))\,a_4\big) \ .
\end{align}
\end{example}

A trace $\tau$ in the sense of Definition~\ref{def:trace} provides a good notion of a state over the nonassociative algebra $\CA$ in the sense of Definition~\ref{def:state}, i.e. $\tau(a^*\,a)\geqslant0$ for all $a\in\CA$ and $\tau(\unit)=1$. Suppose we try to construct further states $\mu$ over $\CA$ from $\tau$, i.e. $\FC$-linear functionals $\mu:\CA\longrightarrow\FC$ satisfying $\mu(a^*\,a)\geqslant0$ for all $a\in\CA$ and $\mu(\unit)=1$. In general, this is not possible: analogously to Remark~\ref{rem:unitvector}, the natural choice would be to take $\mu(a) = \tau\big(\psi^*\,(a\,\psi)\big)$ for a fixed element $\psi\in\CA$, but this fails to be a state because $\tau\big(\psi^*\,((a^*\,a)\,\psi)\big)$ is not generally positive semi-definite.

However $\tau\big((a\,\psi)^*\,(a\,\psi)\big)\geqslant0$, and unravelling this expression using the $2$-cyclicity and $3$-cyclicity properties of $\tau$ expressed by \eqref{eq:4bracketing1} gives
\begin{align}
\tau\big((\psi^*\,a^*)\,(a\,\psi)\big) = \tau\big(\psi^*\,(a^*\,(a\,\psi))\big) = \tau\big(\psi^*\,((\widehat a^*\circ\widehat a)\,\triangleright\,\psi)\big) \ .
\end{align}
This shows that the correct way to generate further states from the trace $\tau$ proceeds through the enveloping algebra $\CCU(\CA)$, via the embedding \smash{$\widehat\CA\subset\CCU(\CA)$}, and it leads to the following general construction.

If $\CA$ is equipped with a trace $\tau$, then there is a distinguished state $\omega_\tau:\CCU(\CA)\longrightarrow\FC$ on the enveloping algebra defined by
\begin{align}
\omega_\tau(A):=\tau(A\,\triangleright\,\unit) \ ,
\end{align}
for all $A\in\CCU(\CA)$. We call $\omega_\tau$ the \emph{tracial state} on $\CCU(\CA)$ with respect to $\tau$. 
Note that $\omega_\tau(\widehat a)=\tau(a)=\omega_\tau(\widehat{a'})$ for all $a\in\CA$, and from \eqref{eq:tauhigher} it follows that $\tau(A\,\triangleright\,\unit) = \tau(A'\,\triangleright\,\unit)$. With a slight abuse of notation, in the following we will denote the {tracial state} simply  by $A\longmapsto\tau(A)$ for $A\in\CCU(\CA)$. 

\begin{remark} \label{rem:losscyclicity}
In traditional quantum mechanics, the state given by the normalized trace is completely mixed: it has maximal ignorance and carries zero information about the system. In our setting, the tracial state $\tau$ does \emph{not} generally define a trace on the enveloping algebra $\CCU(\CA)$, because cyclicity is lost:
\begin{align}
\tau(A\circ B) = \tau\big((A'\,\triangleright\,\unit)\,(B\,\triangleright\,\unit)\big) \ \neq \ \tau\big((B'\,\triangleright\,\unit)\,(A\,\triangleright\,\unit)\big) = \tau(B\circ A) \ ,
\end{align}
for generic $A,B\in\CCU(\CA)$. Hence, by Lemma~\ref{lem:uncert}, uncertainties can  be captured by this state. In particular, $\tau(A^*\circ A)\neq\tau(A\circ A^*)$ generally unless the element $A\in\CCU(\CA)$ is \emph{normal}: $A^*\circ A=A\circ A^*$. This is markedly different to the situation in traditional quantum mechanics.
\end{remark}

More generally, in analogy to the states of Remark~\ref{rem:unitvector}, we can associate a state to any element of the nonassociative algebra $\CA$ through

\begin{proposition}\label{prop:omegapsitau}
Let $\tau$ be a trace on $\CA$, and let $\psi\in\CA$ be a fixed element which is normalized in the sense that $\tau(\psi^*\,\psi)=1$. Then the assignment
\begin{align}
A\longmapsto\omega_\psi(A) := \tau\big(\psi^*\,(A\,\triangleright\, \psi)\big)
\end{align}
defines a state $\omega_\psi\in\mathbb{E}(\CCU(\CA))$.
\end{proposition}

\begin{proof}
We check positivity using $3$-cyclicity of the trace $\tau$ along with Lemma~\ref{lem:leftrightmod} to get
\begin{align}
\begin{split}
\omega_\psi(A^*\circ A) &= \tau\big(\psi^*\,(A^*\,\triangleright\,(A\,\triangleright\,\psi))\big) \\[4pt]
&= \tau\big((A'^*\,\triangleright\,\psi^*)\,(A\,\triangleright\,\psi)\big) = \tau\big((A\,\triangleright\,\psi)^* \, (A\,\triangleright\,\psi)\big) \geqslant  0 \ .
\end{split}
\end{align}
The normalization condition follows easily from $\omega_\psi(\unit) = \tau(\psi^*\,\psi)=1$.
\end{proof}

From a more nonassociative perspective, one can apply this construction to an element $a\in\CA$ using the injection \smash{$\widehat\CA\subset\CCU(\CA)$}. By repeatedly applying $2$-cyclicity as well as $3$-cyclicity, the outcome of measurement in the state $\omega_\psi$ of Proposition~\ref{prop:omegapsitau} can be expressed as
\begin{align}
\langle \widehat a\rangle_\psi := \omega_\psi(\widehat a) = \tau\big(\psi^*\,(a\,\psi)\big) =\tau\big((a\,\psi)\,\psi^*\big) = \tau\big(a\,(\psi\,\psi^*)\big) = \tau(\rho_\psi\,a) \ ,
\end{align}
where the positive semi-definite element
\begin{align}
\rho_\psi := \psi\,\psi^*
\end{align}
of $\CA$ plays the role of a \emph{density matrix} in nonassociative quantum mechanics, that is, $\rho_\psi^\ast = \rho_\psi$ and $\tau(\rho_\psi)=1$. Thus the restriction of the state $\omega_\psi$ to \smash{$\widehat{\CA}\subset\CCU(\CA)$} plays the role of a normal state in the sense of Remark~\ref{rem:normal}.

\begin{remark} \label{rem:mixeddensity}
These constructions can be straightforwardly generalized to ``mixed'' density matrices: let $\{p_1,\dots,p_N\}$ with $p_l\geqslant0$ and $\sum_{l=1}^N\,p_l=1$ be probabilities, and let $\{\psi_1,\dots,\psi_N\}\subset\CA$ be a collection of fixed elements of $\CA$ which are normalized as $\tau(\psi_l^*\,\psi_l)=1$ for each $l=1,\dots,N$. Then the convex linear combinations
\begin{align}
\omega_{\psi}(A) = \sum_{l=1}^N \, p_l\ \tau\big(\psi_l^*\,(A\,\triangleright\,\psi_l)\big) \qquad \text{and} \qquad \rho_{\psi} = \sum_{l=1}^N\,p_l\ \psi_l\,\psi^*_l
\end{align}
satisfy exactly the same properties as above. In the following we shall usually suppress these sums, as all our considerations easily extend to these mixed generalizations.
\end{remark}

\subsection*{\underline{\sf Eigenstates}}

We now give an alternative perspective on the eigenvalues of an element $A\in\CCU(\CA)$.

\begin{definition} \label{def:psieigenvalue}
Let $A\in\CCU(\CA)$. An element $\psi\in\CA$ is an \emph{eigenvector} of $A$ with \emph{eigenvalue} $\lambda\in\FC$ if 
\begin{align}
A\,\triangleright\,\psi = \lambda\,\psi \ .
\end{align}
\end{definition}

If $\psi$ is an eigenvector of $A$ with eigenvalue $\lambda$, then $\psi^*$ is an eigenvector of $A'^*$ with eigenvalue $\overline\lambda$, 
\begin{align}
A'^*\,\triangleright\,\psi^* = \overline\lambda \, \psi^* \ ,
\end{align}
which follows easily from Lemma~\ref{lem:leftrightmod}. This definition is motivated by its relation to Definition~\ref{def:C*eigenvalue} through

\begin{proposition} \label{prop:psiC*ev}
Let $\psi\in\CA$ be a normalised eigenvector of $A\in\CCU(\CA)$ with eigenvalue $\lambda$. Then the associated state $\omega_\psi\in\mathbb{E}(\CCU(\CA))$ is an eigenstate of $A$ with eigenvalue $\lambda$.
\end{proposition}

\begin{proof}
Let $B\in\CCU(\CA)$. Using $A\,\triangleright\,\psi = \lambda\,\psi$ we easily compute
\begin{align}
\omega_\psi(B\circ A) = \tau\big(\psi^*\,(B\,\triangleright\,(A\,\triangleright\,\psi))\big) = \lambda \, \tau\big(\psi^*\,(B\,\triangleright\,\psi)\big) = \lambda\,\omega_\psi(B) \ ,
\end{align}
as required.
\end{proof}

Propositions~\ref{prop:evreal} and~\ref{prop:psiC*ev} together show that observables $O\in\obs(\CCU(\CA))$ also have real eigenvalues in the sense of Definition~\ref{def:psieigenvalue}.

\subsection*{\underline{\sf Tracial GNS construction}}

The proof of Lemma~\ref{lem:CSinequality} can be repeated \emph{verbatum} to show that the map $(\psi_1,\psi_2)\longmapsto\tau(\psi_1^*\,\psi_2)$ is a semi-definite sesquilinear form for any trace $\tau:\CA\longrightarrow\FC$, satisfying
\begin{align}
\tau(\psi_1^*\,\psi_2) = \overline{\tau(\psi_2^*\,\psi_1)}
\end{align}
as well as the Cauchy-Schwartz inequality
\begin{align}
\big|\tau(\psi_1^*\,\psi_2)\big|^2 \leqslant \tau(\psi_1^*\,\psi_1) \, \tau(\psi_2^*\,\psi_2) \ .
\end{align}
In particular, for any $\psi\in\CA$ and $A\in\CCU(\CA)$, the Cauchy-Schwarz inequality together with $3$-cyclicity of the trace imply
\begin{align}\label{eq:CSApsi}
\begin{split}
\tau\big((A\,\triangleright\,\psi)^*\,(A\,\triangleright\,\psi)\big) &= \tau\big((A'^*\,\triangleright\,\psi^*)\,(A\,\triangleright\,\psi)\big) \\[4pt]
&= \tau\big(\psi^*\,(A^*\,\triangleright\,(A\,\triangleright\,\psi))\big) \\[4pt]
&\leqslant \tau(\psi^*\,\psi) \, \tau\big(((A^*\circ A)\,\triangleright\,\psi)^*\,((A^*\circ A)\,\triangleright\,\psi)\big) \ .
\end{split}
\end{align}

Consider the subspace of `zero norm vectors'
\begin{align}
\CJ_\tau := \big\{\psi\in\CA \ \big| \ \tau(\psi^*\,\psi)=0\big\} \ .
\end{align}
If $\psi\in\CJ_\tau$, then from \eqref{eq:CSApsi} it follows that $A\,\triangleright\,\psi\in\CJ_\tau$ for all $A\in\CCU(\CA)$, i.e. $\CCU(\CA)\,\triangleright\,\CJ_\tau\subseteq\CJ_\tau$. This shows that $\CJ_\tau$ is a left ideal of the nonassociative algebra $\CA$ which forms a left submodule over the enveloping algebra $\CCU(\CA)$. Hence the quotient
\begin{align}
\CH_\tau := \CA \, \big\slash \, \CJ_\tau
\end{align}
is a left $\CCU(\CA)$-module on which the trace $\tau$ restricts to produce an inner product.

We denote the equivalence classes in $\CH_\tau$ by
\begin{align}
[a]=\psi_a := \big\{\widetilde a\in\CA \ \big| \ \widetilde a-a \in\CJ_\tau\big\} \ .
\end{align}
On $\CH_\tau$ we define an inner product
\begin{align} \label{eq:tauinnerprod}
(\psi_a,\psi_b)_\tau := \tau(a^*\,b) \ ,
\end{align}
which is well-defined because $\CJ_\tau$ is a left ideal of $\CA$, and makes $\CH_\tau$ into a pre-Hilbert space over $\FC$. We define a linear representation $\pi_\tau:\CCU(\CA)\longrightarrow\sEnd(\CH_\tau)$ by
\begin{align} \label{eq:taurep}
\pi_\tau(A)\,\psi_a := \psi_{A\,\triangleright\, a} \ .
\end{align}

\begin{proposition}
$\pi_\tau$ is a $\ast$-representation of $\CCU(\CA)$ with vacuum vector $\psi_\unit$.
\end{proposition}

\begin{proof}
$\pi_\tau$ is well-defined since $\CH_\tau$ is a left $\CCU(\CA)$-module.
Since $(A\circ B)\,\triangleright\, a = A\,\triangleright\,(B\,\triangleright\, a)$, it is a representation of $\CCU(\CA)$:
\begin{align}
\pi_\tau(A\circ B) = \pi_\tau(A) \, \pi_\tau(B) \ .
\end{align}
It is moreover a $\ast$-representation, as $2$-cyclicity and $3$-cyclicity of the trace $\tau$ imply
\begin{align}
\big(\psi_a,\pi_\tau(A)\,\psi_b\big)_\tau = \tau\big(a^*\,(A\,\triangleright\, b)\big) = \tau\big((A'\,\triangleright\,a^*)\,b\big) = \tau\big((A^*\,\triangleright\, a)^*\,b\big) = \big(\pi_\tau(A^*)\,\psi_a,\psi_b\big)_\tau \ ,
\end{align}
hence
\begin{align}
\pi_\tau(A^*) = \pi_\tau(A)^\dag \ .
\end{align}
Finally, $\psi_\unit$ is a vacuum vector for the representation, as $\psi_a=\psi_{\widehat a\,\triangleright\,\unit}=\pi_\tau(\widehat a)\,\psi_\unit$ with
\begin{align}
\tau(a) = \big(\psi_\unit,\pi_\tau(\widehat a)\,\psi_\unit\big)_\tau \ ,
\end{align}
for all $a\in\CA$.
\end{proof}

\begin{definition}
Let $\tau$ be a trace on the algebra $(\CA,{}^*,\unit)$. The cyclic $\ast$-representation $(\pi_\tau,\psi_\unit)$ is the \emph{tracial GNS representation} of the enveloping algebra $(\CCU(\CA),\circ,{}^\ast,\unit)$ on the pre-Hilbert space $\CH_\tau=\CA\,/\,\CJ_\tau$. 
\end{definition}

\begin{remark}
Compared to the more standard GNS construction of \S\ref{sec:GNS}, the difference in the tracial GNS representation is that its works directly on the nonassociative algebra $\CA$, rather than on the enveloping algebra $\CCU(\CA)$, thus leading to a smaller pre-Hilbert space $\CH_\tau$. If $\CA$ is associative, it can be identified with the usual GNS representation for the tracial state $\tau$ by Example~\ref{lem:bimodcomp} and Lemma~\ref{lem:leftrightmod}.
\end{remark}

In this representation, the state $\omega_a\in\mathbb{E}(\CCU(\CA))$ associated to an element $a\in\CA$ by Proposition~\ref{prop:omegapsitau} is the expectation value
\begin{align}
\omega_a(A) = \big(\psi_a,\pi_\tau(A)\,\psi_a\big)_\tau
\end{align}
of the operator $\pi_\tau(A)\in\sEnd(\CH_\tau)$ in the vector $\psi_a\in\CH_\tau$. Moreover, if $a\in\CA$ is an eigenvector of $A\in\CCU(\CA)$ in the sense of Definition~\ref{def:psieigenvalue}, then $\psi_a\in\CH_\tau$ is an eigenvector of the operator $\pi_\tau(A)$ in the usual sense:
\begin{align}
\pi_\tau(A)\,\psi_a = \psi_{A\,\triangleright\, a} = \psi_{\lambda\,a} = \lambda\,\psi_a \ .
\end{align}
We may then specialise Remark~\ref{rem:linindep} to

\begin{proposition}
If $\{a_1,\dots,a_N\}\subset\CA$ is any collection of normalised eigenvectors of $A\in\CCU(\CA)$ with distinct eigenvalues $\{\lambda_1,\dots,\lambda_N\}\subset\FC$, then  $\{\psi_{a_1},\dots,\psi_{a_N}\}\subset\CH_\tau$ is an orthonormal set:
\begin{align}
(\psi_{a_i},\psi_{a_j})_\tau = \delta_{ij}
\end{align}
for all $i,j\in\{1,\dots,N\}$.
\end{proposition}

\begin{proof}
For each $i,j\in\{1,\dots,N\}$ we compute
\begin{align}
\begin{split}
\big(\psi_{a_i},\pi_\tau(A)\,\psi_{a_j}\big)_\tau &= \lambda_j\,(\psi_{a_i},\psi_{a_j})_\tau \\[4pt]
&= \big(\pi_\tau(A)^\dag\,\psi_{a_i},\psi_{a_j}\big)_\tau = \lambda_i\,(\psi_{a_i},\psi_{a_j})_\tau \ .
\end{split}
\end{align}
This implies $(\lambda_i-\lambda_j)\,(\psi_{a_i},\psi_{a_j})_\tau=0$. The result now follows from $\lambda_i\neq\lambda_j$ for $i\neq j$ together with $(\psi_{a_i},\psi_{a_i})_\tau = \tau(a_i^*\,a_i)=1$.
\end{proof}

\subsection*{\underline{\sf GNS-Dirac representation}}

The tracial GNS representation can be used to make contact with the more familiar looking expressions in the physics literature. For this, we use the \emph{Dirac notation} for the inner product \eqref{eq:tauinnerprod} to express it as
\begin{align}
\langle \psi_a|\psi_b\rangle := (\psi_a,\psi_b)_\tau \ ,
\end{align}
and we formally write
\begin{align}
\Tr_{\CH_\tau}\big(|\psi\rangle\langle\psi|\,\pi_\tau(O)\big) := \langle\psi|\,\pi_\tau(O)\,|\psi\rangle = \big(\psi,\pi_\tau(O)\,\psi\big)_\tau = \tau\big(\psi^*\,(O\,\triangleright\,\psi)\big) 
\end{align}
for $O\in\obs(\CCU(\CA))$ and $\psi\in\CH_\tau$, where $\pi_\tau$ is the representation \eqref{eq:taurep} of the enveloping algebra $\CCU(\CA)$ as operators on the pre-Hilbert space $\CH_\tau=\CA/\CJ_\tau$. With an abuse of notation, here we identify $\psi\in\CA$ with its equivalence class $[\psi]\in\CH_\tau$ to simplify the presentation. 

The density matrices of Remark~\ref{rem:mixeddensity} may then be represented as
\begin{align} \label{eq:pitaurhopsi}
\pi_\tau(\widehat\rho_\psi) = \sum_{l=1}^N\,p_l \ |\psi_l\rangle\langle\psi_l| \ .
\end{align}
The corresponding states thereby read
\begin{align}
\omega_\psi(O) = \Tr_{\CH_\tau}\big(\pi_\tau(\widehat\rho_\psi)\,\pi_\tau(O)\big) = \sum_{l=1}^N\,p_l \ \big(\psi_l,\pi_\tau(O)\,\psi_l\big)_\tau = \sum_{l=1}^N\,p_l \ \tau\big(\psi_l^*\,(O\,\triangleright\,\psi_l)\big) \ .
\end{align}
This shows that the states $\omega_\psi$, in this representation, are normal states in the sense of Remark~\ref{rem:normal}.

\section{Dynamics}
\label{sec:CPM}

Let us now take a look at equations of motion in nonassociative quantum mechanics from our algebraic perspective.

\subsection*{\underline{\sf Completely positive maps}}

Time evolution of states in quantum mechanics is captured by the notion of a `completely positive map'. Let $\CA$ be a $\FC$-algebra with enveloping algebra $\CCU(\CA)$.

\begin{definition}
A map $\phi:\mathbb{E}(\CCU(\CA))\longrightarrow \mathbb{E}(\CCU(\CA))$ from states to states is a \emph{positive map}. 
A positive map $\phi$ is a \emph{completely positive map} if it can be consistently extended to a positive map $\phi\otimes\text{id}$ on $\mathbb{E}(\CCU(\CA)\otimes\CCU(\CB))$ for any unital $\ast$-algebra $\CB$.
\end{definition}

\begin{remark}
The tensor product $\CCU(\CA)\otimes\CCU(\CB)$ represents the coupling of a system to its environment. In traditional quantum mechanics, the existence of a completely positive map is non-trivial due to entanglement of states. For example, the transposition of density matrices is a positive map but not a completely positive map. 
\end{remark}

Let $\{A_1,\dots,A_N\}$ be a collection of elements of $\CCU(\CA)$ which are \emph{normalized} in the sense that
\begin{align} \label{eq:Krausssum}
\sum_{k=1}^N\,A_k^\ast \circ A_k = \unit \ .
\end{align}
The elements $A_k$ are the analogs of \emph{Krauss operators} in traditional quantum mechanics: in that case any completely positive map can be expressed in terms of Krauss operators by Choi's Theorem. 
In the nonassociative setting the normalization condition would be hard to satisfy for elements of $\CA$, but not in $\CCU(\CA)$, as we will see below. In particular, we do not impose an additional unital property that Krauss operators are sometimes assumed to have.

Given $\{A_1,\dots,A_N\}$ we define a completely positive map $\omega\longmapsto\widetilde\omega$ by setting
\begin{align} \label{eq:Kraussomega}
\widetilde\omega(O) := \sum_{k=1}^N\,\omega\big(A_k^\ast\circ O \circ A_k\big) \ ,
\end{align}
for any state $\omega\in\mathbb{E}(\CCU(\CA))$ and observable $O\in\obs(\CCU(\CA))$. It is easy to check $\widetilde\omega$ is positive and normalized. This has the meaning of a \emph{quantum channel} that carries quantum information. We may interpret this map in the Heisenberg picture as a map of observables 
\begin{align} \label{eq:obsmap}
\widetilde O = \sum_{k=1}^N\,A_k^\ast\circ O \circ A_k \qquad \text{with} \quad \widetilde\unit = \sum_{k=1}^N\,A_k^\ast \circ A_k = \unit \ .
\end{align}

\begin{example}
We observe that all states of Remark~\ref{rem:mixeddensity} arise as completely positive maps from the tracial state $\tau$: for an observable $O\in\obs(\CCU(\CA))$, setting $A_l=\sqrt{p_l} \ \widehat\psi_l$ yields
\begin{align}
\widetilde\tau(O) = \sum_{l=1}^N \, p_l\,\tau\big((\widehat\psi_l^*\circ O\circ\widehat\psi_l)\,\triangleright\,\unit\big) = \sum_{l=1}^N\,p_l\,\tau\big(\psi_l^*\,(O\,\triangleright\,\psi_l)\big) = \omega_{\psi}(O) \ .
\end{align}
\end{example}

\begin{example} \label{ex:CPTcompose}
We may also compose the states of Remark~\ref{rem:mixeddensity} with any other completely positive map. We illustrate this for the state
\begin{align}
\omega_\psi(O) = \tau\big(\psi^\ast\,(O\,\triangleright\, \psi)\big) \ .
\end{align} 
By repeatedly applying the cyclicity properties of the trace $\tau:\CA\longrightarrow\FC$, and using Lemma~\ref{lem:leftrightmod}, we can compute the new state \smash{$\widetilde\omega_\psi$} as
\begin{align}
\begin{split}
\widetilde\omega_\psi(O) &= \sum_{k=1}^N\,\tau\big(\psi^\ast\,(A_k^\ast\,\triangleright\,(O\,\triangleright\,(A_k\,\triangleright\,\psi)))\big) \\[4pt]
&= \sum_{k=1}^N\,\tau\big((A_k'^*\,\triangleright\,\psi^*)\,(O\,\triangleright\,(A_k\,\triangleright\,\psi))\big) = \sum_{k=1}^N\,\tau\big((A_k\,\triangleright\,\psi)^*\,(O\,\triangleright\,(A_k\,\triangleright\,\psi))\big) \ .
\end{split}
\end{align}

Restricting this state to the vector space embedding \smash{$\widehat\CA\subset\CCU(\CA)$} shows that the density matrix $\rho_\psi=\psi\,\psi^\ast$ transforms to the new density matrix
\begin{align} \label{eq:Kraussrho}
 \widetilde\rho_\psi = \sum_{k=1}^N\,(A_k\,\triangleright\,\psi)\,(A_k\,\triangleright\,\psi)^* \ ,
\end{align}
which we interpret as a map in the Schr\"odinger picture.
\end{example}

\subsection*{\underline{\sf Unitary time evolution}}

Fix an observable $H=H^*$ in $\obs(\CCU(\CA))$, which we identify as the Hamiltonian of a given quantum system. For a time parameter $t\in\FR$, the exponentials
\begin{align} \label{eq:unitaries}
U_t = \exp_\circ\big(-\tfrac\ii\hbar\,t\,H\big)
\end{align}
are defined by computing powers in the Taylor series of the exponential using the composition product. We assume that $H$ itself is time-independent (otherwise $U_t$ should be defined using a path-ordered exponential). We further assume that $U_t$ are elements of $\CCU(\CA)$ (for $C^*$-algebras this would be automatically guaranteed by standard functional calculus). 

The one-parameter family \eqref{eq:unitaries} forms a faithful representation of the additive group $\FR$ in $\sEnd(\CCU(\CA))$:
\begin{align}
U_{t+t'} = U_t\circ U_{t'}
\end{align}
for $t,t'\in\FR$. In particular, this implies that $U_t$ are unitaries, i.e. they
satisfy
\begin{align}
U_t^\ast\circ U_t = \unit = U_t\circ U_t^* \ ,
\end{align}
where $U_t^*=U_{-t}$ and $U_0=\unit$.

The unitaries $U_t$ provide a quantum channel with completely positive map $\omega\longmapsto\omega_t$. That is, given a state $\omega\in\mathbb{E}(\CCU(\CA))$, let
\begin{align}
\omega_t(O) := \omega(U_t^\ast\circ O \circ U_t) \ ,
\end{align}
with the initial condition $\omega_0=\omega$. We assume that the functions $t\longmapsto\omega_t(O)$ are continuous for every $O\in\obs(\CCU(\CA))$.

The infinitesimal generator of this completely positive map is then given by
\begin{align}
\frac{\dd\,\omega_t(O)}{\dd t}:= \lim_{t'\to0} \, \frac{\omega_{t+t'}(O)-\omega_t(O)}{t'} = \frac\ii\hbar \, \omega_t\big([H,O]_\circ \big) \ ,
\end{align}
provided the limit exists for every $O\in\obs(\CCU(\CA))$.
We interpret this equation as the Heisenberg equation of motion
\begin{align} \label{eq:dOdt}
\frac{\dd O}{\dd t} = \frac\ii\hbar \, [H,O]_\circ
\end{align}
in nonassociative quantum mechanics. Note that the right-hand side is a derivation (in $O$) of the enveloping algebra $\CCU(\CA)$, and hence is compatible with the time derivative on the left-hand side.

For the transformed density matrix
\begin{align}
\rho_\psi^{t} = (U_t\,\triangleright\,\psi)\,(U_t\,\triangleright\,\psi)^* \ ,
\end{align}
with the initial condition $\rho^0_\psi = \rho_{\psi}$,  we need to endow the algebra $\CA$ with the additional structure of a metric space. This will be the case for the explicit examples we consider later on. We assume that the maps $t\longmapsto\rho_\psi^t$ are continuous for every $\psi\in\CA$. 

We then find
\begin{align} \label{eq:drhodt}
\ii\,\hbar\,\frac{\dd\rho_\psi^{t}}{\dd t} := \ii\,\hbar\, \lim_{t'\to0} \, \frac{\rho_\psi^{t+t'} - \rho_\psi^t}{t'} = (H\,\triangleright\,\psi)\,\psi^* - \psi\,(H'\,\triangleright\,\psi^*) \ ,
\end{align}
provided the limit exists for every $\psi\in\CA$.
We interpret this equation as following from the nonassociative Schr\"odinger equation
\begin{align} \label{eq:Schrodinger}
\ii\,\hbar\,\frac{\dd\psi}{\dd t} = H\,\triangleright\, \psi 
\end{align}
involving the nonassociative product of the algebra $\CA$. In particular, this demonstrates that the eigenvalue problem should be studied for the element $\psi\in\CA$ and \emph{not} its associated density matrix $\rho_\psi = \psi\,\psi^*$, contrary to the treatment of~\cite{Mylonas:2013jha}.

\subsection*{\underline{\sf Master equation}}

Let us now extend this description of time evolution by deriving the analog of the Lindblad master equation for an open system in nonassociative quantum mechanics. The technical details are rather involved but standard, so here we only provide a heuristic sketch of the derivation. 

We start from the completely positive map $\omega\longmapsto\widetilde\omega$ given by
\begin{align} \label{eq:omegaprimeO}
\widetilde\omega(O) = \sum_{k=0}^N\, \omega(A_k^*\circ O\circ A_k) \ ,
\end{align}
and we set
\begin{align}
A_0=\unit+L_0\,\delta t-\tfrac\ii\hbar\,H\,\delta t \qquad \text{and} \qquad A_k=L_k\,\sqrt{\delta t} \ ,
\end{align}
for $k=1,\dots,N$. Here $L_0=L_0^*$ and $H=H^\star$ are observables in $\obs(\CCU(\CA)$), while $L_1,\dots,L_N$ are \emph{jump operators} in $\CCU(\CA)$ encoding dissipation into the dynamics, and $\delta t>0$ is an infinitesimal variation of the time parameter $t\in\FR$. Inserting this into \eqref{eq:omegaprimeO} yields
\begin{align}
\frac{\dd\,\omega_t(O)}{\dd t} = \omega\Big(\{L_0,O\}_\circ + \tfrac\ii\hbar\,[H,O]_\circ + \mbox{$\sum\limits_{k=1}^N$}\,L_k^*\circ O\circ L_k\Big) \ .
\end{align}

Imposing the condition 
\begin{align}
\sum_{k=0}^N\,A_k^*\circ A_k = \unit
\end{align}
then enables us to solve for $L_0$ in terms of $L_k$ for $k=1,\dots,N$ as
\begin{align}
L_0 = -\frac12\,\sum_{k=1}^N\,L_k^*\circ L_k \ .
\end{align}
This yields
\begin{align}
\frac{\dd\,\omega_t(O)}{\dd t} = \frac\ii\hbar\,\omega\big([H,O]_\circ\big) + \sum_{k=1}^N\,\omega\big(L_k^*\circ O\circ L_k-\tfrac12\,\{L_k^*\circ L_k,O\}_\circ\big) \ .
\end{align}
We may interpret this evolution equation in the Heisenberg picture as the equation of motion
\begin{align} \label{eq:mastereq}
\frac{\dd O}{\dd t} = \frac\ii\hbar\,[H,O]_\circ + \sum_{k=1}^N\, L_k^*\circ O\circ L_k-\frac12\,\{L_k^*\circ L_k,O\}_\circ
\end{align}
for each quantum observable $O\in\obs(\CCU(\CA))$.

\subsection*{\underline{\sf Nonassociative dynamics}}

To capture a more nonassociative perspective on dynamics which relies less on the full associative enveloping algebra, along the lines of \S\ref{sec:trace}, we suppose that $\CA$ is endowed with a trace $\tau$ in the sense of  Definition~\ref{def:trace}. We take the Krauss operators to be elements $a_k\in\CA$ for $k=1,\dots,N$ and use them to construct completely positive dynamics for the positive observable $b^*\,b$ with $b\in\CA$. Given an element $\psi\in\CA$,  there are then two possible choices for  the completely positive dynamics of the tracial state of $\CA$, reflecting the fact that nonassociativity breaks operator-state duality.

Consider first
\begin{align}
\begin{split}
0\leqslant\mu_\psi(b^*\,b) :\!\!&= \sum_{k=1}^N\,\tau\big(((\psi^*\,a_k^*)\,b^*)\,(b\,(a_k\,\psi))\big) \\[4pt]
&= \sum_{k=1}^N\,\tau\big((\psi^*\,a_k^*)\,(b^*\,(b\,(a_k\,\psi)))\big) = \sum_{k=1}^N\,\tau\big((\widehat a_k\,\triangleright\,\psi)^*\,((\,\widehat b^*\circ \widehat b\,)\,\triangleright\,(\widehat a_k\triangleright\,\psi))\big) \ .
\end{split}
\end{align}
This choice is just a special case of the construction from Example~\ref{ex:CPTcompose} for $O=\widehat b^*\circ\widehat b\in\obs(\CCU(\CA))$ and $A_k=\widehat a_k\in\CCU(\CA)$.

Alternatively, by using a different placement of parantheses we can consider the \emph{dual state} with
\begin{align}
\begin{split}
0\leqslant \mu_\psi^{\textrm{\tiny$\vee$}}(b^*\,b) :\!\!&= \sum_{k=1}^N\, \tau\big((\psi^*\,(a_k^*\,b^*))\,((b\,a_k)\,\psi)\big) \\[4pt]
&= \sum_{k=1}^N\,\tau\big((b\,a_k)\,(\psi\,(\psi^*\,(a_k^*\,b^*))\big) \\[4pt]
&= \sum_{k=1}^N\,\tau\big(b\,(a_k\,(\psi\,(\psi^*\,(a_k^*\,b^*)))\big) = \sum_{k=1}^N\,\tau\big(\,\widehat b\circ\widehat a_k\circ(\widehat\psi\circ\widehat\psi^*)\circ\widehat a_k^*\circ\widehat b^*\big) \ .
\end{split}
\end{align}
This equation tells us two things. Firstly, it identifies the quantity $\widehat\rho_{\widehat\psi}:=\widehat\psi\circ\widehat\psi^*$ as the analogue of a positive semi-definite density matrix, which in the dual picture lives in the enveloping algebra $\CCU(\CA)$ rather than in $\CA$. Secondly, the constitutents $\widehat b$ and $\widehat b^*$ of the observable $O=\widehat b^*\circ\widehat b$ are \emph{split} in the dual state. 

Split operators also appear in alternative proposals for nonassociative dynamics that replace commutators by associators in the Heisenberg equation of motion, leading symbolically to expressions such as
\begin{align} \label{eq:H1H2}
\frac{\dd O}{\dd t} = [H_1,O,H_2] := (H_1\,O)\,H_2 - H_1\,(O\,H_2)
\end{align}
involving a pair of Hamiltonians $(H_1,H_2)$~\cite{Biedenharn:1981pa,Liebmann19}. This expression can be corrected in order to preserve self-adjointness, and hence reality of probabilities. The corresponding dual Schr\"odinger-Liouville equation is trace-preserving, as the associator vanishes upon acting with a $3$-cyclic trace, thus leading to normalized probabilities. But there does not seem to be a natural notion of positivity, in constrast to our formalism, and hence these theories have no probabilistic interpretation. Dynamics such as \eqref{eq:H1H2} should be regarded as a quantum version of Nambu's \emph{generalized dynamics}~\cite{Nambu:1973qe,Mylonas:2013jha}, but not as nonassociative quantum mechanics.

\subsection*{\underline{\sf GNS-Dirac representation}}

Let us now collect the main dynamical formulas of this section in the tracial GNS representation of \S\ref{sec:trace} in terms of operators on the pre-Hilbert space $\CH_\tau = \CA/\CJ_\tau$. In the GNS-Dirac notation, the normalization condition \eqref{eq:Krausssum} for Krauss operators becomes
\begin{align}
\sum_{k=1}^N\,\pi_\tau(A_k)^\dag\,\pi_\tau(A_k) = \pi_\tau(\unit) = \id_{\CH_\tau} \ ,
\end{align}
while the transformation \eqref{eq:Kraussrho} of the density matrix $\pi_\tau(\rho_\psi)$ given by \eqref{eq:pitaurhopsi} reads
\begin{align}
\pi_\tau(\,\widehat{\widetilde\rho}_\psi) = \sum_{k=1}^N\,\pi_\tau(A_k) \, \pi_\tau(\widehat\rho_\psi)\,\pi_\tau(A_k)^\dag \ .
\end{align}
Similarly, the map of observables $O\in\obs(\CCU(\CA))$ given by \eqref{eq:obsmap} is represented as
\begin{align}
\pi_\tau(\widetilde O) = \sum_{k=1}^N\,\pi_\tau(A_k)^\dag \, \pi_\tau(O) \, \pi_\tau(A_k) \ .
\end{align}
These are the more familiar looking formulas.

The time evolution equation \eqref{eq:drhodt} now assumes the familiar form of a Schr\"odinger-Liouville equation
\begin{align}
\frac{\dd\,\pi_\tau(\widehat\rho_\psi^{\,t})}{\dd t} = \frac\ii\hbar \, \big[\pi_\tau(\widehat\rho_\psi),\pi_\tau(H)\big] \ ,
\end{align}
whereas \eqref{eq:dOdt} takes the standard form of a Heisenberg equation of motion
\begin{align}
\frac{\dd\,\pi_\tau(O)}{\dd t} = \frac\ii\hbar\,\big[\pi_\tau(H),\pi_\tau(O)\big] \ .
\end{align}
The Lindblad equations \eqref{eq:mastereq} simplify in a similar way to
\begin{align}
\frac{\dd \,\pi_\tau(O)}{\dd t} = \frac\ii\hbar\,\big[\pi_\tau(H),\pi_\tau(O)\big] + \sum_{k=1}^N\, \pi_\tau(L_k)^\dag\, \pi_\tau(O)\, \pi_\tau(L_k)-\frac12\,\big\{\pi_\tau(L_k)^\dag\, \pi_\tau(L_k),\pi_\tau(O)\big\} \ .
\end{align}

\section{Jordan algebras}
\label{sec:Jordan}

As our first and simplest explicit example, let $\CA$ be a noncommutative associative unital $*$-algebra over $\FC$. On the underlying vector space of $\CA$ we define a commutative nonassociative product by
\begin{align} \label{eq:Jordanprod}
a \, b := \tfrac12\,(a\cdot b+b\cdot a) \ ,
\end{align}
for $a,b\in\CA$, where the dot denotes the original multiplication on $\CA$. Nonassociativity is easily seen through the computation
\begin{align} \label{eq:JordanNA}
a\,(b\,c) - (a\,b)\, c = \tfrac14\,[[a,c]_\cdot,b]_\cdot \ ,
\end{align}
where $[a,c]_\cdot = a\cdot c - c\cdot a$ is the commutator in the original algebra; this is non-vanishing in general.
This new product is called the \emph{Jordan product}, and it makes $\CA$ into a commutative \emph{Jordan algebra}; see e.g.~\cite{Szabo:2019hhg} for a review and references. 

Here we only consider the case where $\CA=\MM_n(\FC)$ is the finite-dimensional algebra of $n{\times}n$ complex-valued matrices. The unit element $\unit$ is the $n{\times}n$ identity matrix, while the natural $*$-structure is complex conjugation followed by matrix transposition ${}^\top$:
\begin{align}
a^* := a^\dag = \overline a{}^\top \ .
\end{align}
The matrix algebra $\MM_n(\FC)$ is generated by $\unit$ and the $n{\times}n$ matrix units $e_{ij}$, with $1\leqslant i,j\leqslant n$, i.e.~$e_{ij}$ is the $n{\times}n$ matrix whose $(i,j)$-th entry is equal to $1$ and has zeroes for all other entries. They obey the relations $e_{ij}\cdot e_{kl} = \delta_{jk} \, e_{il}$, where the dot denotes ordinary matrix multiplication. Hence the relations of the Jordan algebra are
\begin{align}
e_{ij}\,e_{kl} = \tfrac12\, (\delta_{jk}\,e_{il} + \delta_{il}\,e_{kj}) \qquad \mbox{with} \quad \sum_{i=1}^n\,e_{ii} = \unit \ .
\end{align}

\subsection*{\underline{\sf Enveloping algebra}}

Let us explicitly describe the enveloping algebra $\CCU(\MM_n(\FC))$, acting in the regular birepresentation, in this case. By Example~\ref{lem:bimodcomp2}, it suffices to consider, say, the left module structure on $\CA$. Under the isomorphism $\mathbbm{M}_n(\FC)\simeq\FC^n\otimes\FC^n$, the left module structure on $\MM_n(\FC)$ defines a linear representation $\varPi:\mathbbm{M}_n(\FC)\longrightarrow\mathbbm{M}_n(\FC)\otimes\mathbbm{M}_n(\FC)$ given by
\begin{align}
\varPi(a)\cdot x := \widehat a\,\triangleright\, x = a\,x \ ,
\end{align}
for $a,x\in\MM_n(\FC)$. Explicitly
\begin{align}
\varPi(a) = \tfrac12\,\big(a\otimes\unit + \unit\otimes a^\top\big) \ .
\end{align}
Then the composition product is given by matrix multiplication:
\begin{align} \label{eq:Jordancomp}
\widehat a\circ \widehat b := \varPi(a) \cdot \varPi(b) = \tfrac14\,\big((a\cdot b)\otimes\unit + a\otimes b^\top + b\otimes a^\top + \unit\otimes(b\cdot a)^\top\big) \ ,
\end{align}
where we note the ordering in the last term. The unit of $\CCU(\MM_n(\FC))$ is $\varPi(\unit)=\unit\otimes\unit$ and the $*$-involution is the tensor product ${}^* = {}^\dag\otimes{}^\dag$. With this description we can straightforwardly establish

\begin{proposition} \label{prop:Jordancompalg}
The enveloping algebra of  the Jordan algebra $\MM_n(\FC)$ is the matrix algebra
\begin{align}
\CCU\big(\mathbbm{M}_n(\FC)\big) = \mathbbm{M}_{n^2}(\FC) \ .
\end{align}
\end{proposition}

\begin{proof}
We use the isomorphism $\MM_{n^2}(\FC)\simeq\MM_n(\FC)\otimes\MM_n(\FC)$ of vector spaces. The generators of the enveloping algebra $\CCU(\MM_n(\FC))$ in this basis are given by the unit $\unit\otimes\unit$ together with
\begin{align}
\varPi_{ij} := \varPi(e_{ij}) = \tfrac12\,(e_{ij}\otimes\unit + \unit\otimes e_{ji}\big) \ ,
\end{align}
for $1\leqslant i,j\leqslant n$. We show that these generate all tensor products $e_{ij}\otimes e_{kl}$. 

This can be done iteratively starting from
\begin{align}
E_{ij} := e_{ij}\otimes e_{ij}
\end{align}
for $1\leqslant i,j\leqslant n$. 
These are given in terms of the generators $\varPi_{ij}$ through
\begin{align}
E_{ii} = 2\, \varPi_{ii}\cdot\varPi_{ii} - \varPi_{ii} \qquad \text{and} \qquad E_{ij} = 4\,E_{ii} \cdot \varPi_{ij}\cdot \varPi_{ji}\cdot E_{jj} \ ,
\end{align}
where here and below we assume that $i\neq j\neq k\neq l$. The remaining 12 possible index structures can now all be expressed entirely in terms of these combinations; for example
\begin{align}
e_{ij}\otimes e_{ii} = 2\,E_{ii}\cdot\varPi_{ij} \qquad \text{and} \qquad e_{ij}\otimes e_{kl} = 4\,\varPi_{ik}\cdot E_{kj}\cdot \varPi_{lj} \ ,
\end{align}
and so on.
\end{proof}

\subsection*{\underline{\sf Positive elements}}

The natural tracial state over the Jordan algebra $\MM_n(\FC)$ is provided by the normalized matrix trace:
\begin{align}
\tau(a) := \tfrac1n \, \Tr(a) \ .
\end{align}
By cyclicity of the matrix trace it follows that $\Tr(a\,b)=\Tr(a\cdot b)$. From this it follows trivially that $\tau$ is positive and $2$-cyclic with respect to the Jordan product. One also easily checks that $\tau$ is $3$-cyclic, as required.

Positive elements $a$ with respect to the trace $\tau$ are given by the elements
\begin{align}
a^*\,a = \tfrac12\,\big(a^\dag\cdot a + a\cdot a^\dag\big) =: c \ ,
\end{align}
where $c$ is a positive semi-definite matrix. It follows that $a$ can be replaced by a unique square root $\sqrt{c}$. In particular, for the density matrices
\begin{align}
\rho_\psi = \psi\,\psi^* = \tfrac12\,\big(\psi\cdot\psi^\dag + \psi^\dag\cdot\psi\big) =: \varphi\cdot\varphi \ .
\end{align}
This enables one to redefine the corresponding state $\omega_\psi$ to
\begin{align}
\omega_\varphi(A) = \tau\big(\varphi\,(A\,\triangleright\,\varphi)\big) \ ,
\end{align}
in terms of a positive semi-definite matrix $\varphi$. From this perspective, pure states correspond to projectors, that is, $\varphi^2 = \varphi$.

\subsection*{\underline{\sf Tracial GNS representation}}

For the tracial state $\tau$ the ideal of zero norm elements is trivial:
\begin{align}
\CJ_\tau = \big\{a\in\MM_n(\FC) \ \big| \ \Tr(a^\dag\,a)=0\big\} = \{0\} \ ,
\end{align}
and hence the tracial GNS Hilbert space can be identified with $\MM_n(\FC)\simeq\FC^{n^2}$ as a vector space:
\begin{align}
\CH_\tau = \FC^{n^2} \ ,
\end{align}
with the standard inner product $(\,\cdot\,,\,\cdot\,)_\tau$.
By Proposition~\ref{prop:Jordancompalg}, the tracial GNS representation is the fundamental representation of the enveloping algebra $\CCU(\MM_n(\FC)) = \MM_{n^2}(\FC)$:
\begin{align}
\pi_\tau(A)\,\psi = A\cdot \psi \ ,
\end{align}
for $A\in\MM_{n^2}(\FC)$ and $\psi\in\FC^{n^2}$.

\subsection*{\underline{\sf Eigenstates}}

Let us now look at eigenstates in the Jordan algebra, which can be characterised according to
\begin{proposition} \label{prop:eigenJordan}
For $a\in\MM_n(\FC)$, $\lambda\in\FC$ and a positive semi-definite $n{\times}n$ matrix $\varphi$ of rank $r$, the eigenvalue equation 
\begin{align}
a\,\varphi =  \lambda\,\varphi
\end{align}
is solved by the spectral decomposition
\begin{align} \label{eq:spectralrep}
\varphi = \sum_{i=1}^N \, \sqrt{p_i} \ \phi_i\otimes\phi_i^\dag \ ,
\end{align}
for some $N\in\mathbbm{N}$, where $p_i>0$ with $\frac1r\,\sum_{i=1}^N\,p_i=1$, and $\phi_i\in\FC^n$ are ordinary eigenvectors of the matrix $a$ with fixed eigenvalue $\lambda$, i.e. $a \cdot \phi_i = \lambda\, \phi_i$. 
\end{proposition}

\begin{proof}
The eigenvalue equation reads
\begin{align}
\tfrac12\,(a\cdot\varphi + \varphi\cdot a) = \lambda\,\varphi \ ,
\end{align}
which implies that the matrices $a-\lambda\,\unit$ and $\varphi$ anticommute. 
The matrix $\varphi$ has a spectral representation 
\begin{align}
\varphi = \sum_{i=1}^N\,\alpha_i \ \phi_i\otimes\phi_i^\dag \ ,
\end{align}
for some $N\in\mathbbm{N}$, with $\alpha_i>0$ and orthonormal eigenvectors $\phi_i\in\FC^n$, i.e. $\varphi \cdot \phi_i = \alpha_i \, \phi_i$. With the normalization $\tau(\varphi^2) = 1$, this implies $\alpha_i=\sqrt{p_i}$ with $p_i>0$ and $\frac1r\,\sum_{i=1}^N\,p_i=1$.
Then
\begin{align}
\varphi\cdot(a-\lambda\,\unit)\cdot\phi_i = -(a-\lambda\,\unit)\cdot\varphi\cdot\phi_i = -\sqrt{p_i} \, (a-\lambda\,\unit)\cdot\phi_i \ .
\end{align}
Since $-\sqrt{p_i}<0$ and $\varphi$ is positive semi-definite, it follows that $a\cdot\phi_i=\lambda\,\phi_i$.
\end{proof}

From the perspective of Proposition~\ref{prop:eigenJordan}, pure states are of the form $\varphi = \phi_i\otimes\phi^\dag_i$ (no sum on $i$), with density matrix $\rho_\psi=\varphi^2=\varphi$.

\subsection*{\underline{\sf Dynamics}}

The original motivation for the introduction of the Jordan product of operators was that it maps observables to observables~\cite{Jordan32,Jordan:1933vh}: The Jordan product of Hermitean matrices is again Hermitean. More generally any real polynomial in observables is again an observable. In the following we consider a special class of examples.

Starting from two observables $y$ and $z$ in the Jordan algebra $\CA = \MM_n(\FC)$, we can construct a \emph{bona fide} nonassociative observable 
\begin{align}
H = \hbar\,\varpi\, \reallywidehat{\big(z\,(y\,y)-(z\,y)\,y\big)} \label{eq:bonafide}
\end{align}
in the enveloping algebra $\CCU(\CA)=\MM_{n^2}(\FC)$ as Hamiltonian, where $\hbar\,\varpi$ is a constant with units of energy. In traditional associative quantum mechanics this expression would of course be zero, but in nonassociative quantum mechanics it is non-trivial in general. Using the Jordan product \eqref{eq:Jordanprod}, it can be checked that $H=H^*$. In fact, using \eqref{eq:JordanNA} we can write the Hamiltonian in the form
\begin{align} \label{eq:JordanH}
H = \tfrac14\,\hbar\,\varpi\,\reallywidehat{[[z,y]_\cdot,y]_\cdot} \ .
\end{align}

\begin{example} \label{ex:Jordan1}
Let $(x,y,z) = (\sigma_1,\sigma_2,\sigma_3)$ be the $2{\times}2$ Pauli spin matrices. They satisfy the defining relations of the Lie algebra $\mathfrak{su}(2)$ and of the Clifford algebra ${\sf Cl}(3)$ in three dimensions:
\begin{align} \label{eq:sigmarels}
[\sigma_i,\sigma_j]_\cdot = 2\,\ii\,\epsilon_{ijk}\,\sigma_k
\qquad \text{and} \qquad 
\{\sigma_i,\sigma_j\}_\cdot = 2\,\delta_{ij}\,\unit \ .
\end{align}
Let $\CA=\MM_2(\FC)$ be the Jordan algebra generated by $\sigma_1$, $\sigma_2$, $\sigma_3$ and the unit $\unit$. From \eqref{eq:Jordanprod} and \eqref{eq:sigmarels} it follows that the Jordan product sees only the Clifford algebra structure of the Pauli matrices through 
\begin{align} \label{eq:Jordansigma}
\sigma_i\,\sigma_j = \delta_{ij}\,\unit
\end{align}
in $\CA$, while it is blind to the Lie algebra structure. 

From \eqref{eq:JordanH} the Hamiltonian is
\begin{align} \label{eq:HhatZ}
H = \hbar\,\varpi\, \widehat z = \tfrac12\,\hbar\,\varpi\,\big(z\otimes\unit + \unit\otimes z^\top\big) \ .
\end{align}
Its action on a general element
\begin{align}
\psi = a\,\unit + b\,x + c\,y + d\,z
\end{align}
of $\MM_2(\FC)$, with $a,b,c,d\in\FC$, can be computed from \eqref{eq:Jordansigma} which yields the relations $z\,\unit=z$, $z\,x=0$,  $z\,y=0$ and $z\,z=\unit$. 

We find
\begin{align}
\begin{split}
H\,\triangleright\,\psi &= \hbar\,\varpi\,z\,\psi \\[4pt]
&= \hbar\,\varpi\,(a\,z + d\,\unit) = \ii\,\hbar\,(\dot a\,\unit + \dot b\, x + \dot c\,y + \dot d\,z) \ ,
\end{split}
\end{align}
where an overdot represents time derivative and we used the Schr\"odinger equation \eqref{eq:Schrodinger}. The solution is given by
\begin{align}
a(t) = A\cos(\varpi\,t+ \phi) \quad , \quad b(t) = B \quad , \quad c(t) = C \quad , \quad d(t) = -\ii\,A\sin(\varpi\,t+\phi)
\end{align}
with constants $A,B,C,D\in\FC$ and $\phi\in[0,2\pi)$.
\end{example}

Even though the Jordan algebra $\CA$ is commutative and hence its commutators vanish, the Heisenberg equation \eqref{eq:dOdt} is non-trivial. The reason is of course that the associative enveloping algebra $\CCU(\CA)$ is noncommutative. Hence in our approach there is  no need to consider generalized Heisenberg equations of motion such as \eqref{eq:H1H2}. 

In particular, consider \eqref{eq:dOdt} with Hermitean Hamiltonian $H=\widehat h$ and an observable $O=\widehat{o} \in\obs(\CCU(\CA))$ which are images of elements $h,o\in\CA$. We need the operator commutator, which we can get from \eqref{eq:Jordancomp} as
\begin{align}
[H,O]_\circ = \tfrac14\, \big([h,o]_\cdot\otimes\unit + \unit\otimes[h^\top,o^\top]_\cdot\big) \ .
\end{align}
Note that the transposes of observables are also observables. 

\begin{example}
Let us again choose \eqref{eq:HhatZ} as Hamiltonian. From the Heisenberg equation  \eqref{eq:dOdt} we find that the observable $\widehat z$ is a constant of the motion and that the observables $a\, \widehat x + b\,\widehat y$ for $a,b\in\FR$ lead to oscillatory behaviour analogously to Example~\ref{ex:Jordan1}. The details are left as an easy exercise for the reader.
\end{example}

\section{Octonions}
\label{sec:oct}

A richer example is provided by the eight-dimensional (complex) nonassociative algebra of octonions $\CA=\mathbbm{O}$. It is generated by the real unit $e_0=\unit$ and the seven imaginary unit octonions $e_i$ with the relations
\begin{align} \label{eq:octrels}
e_i\,e_j = -\delta_{ij}\,e_0 + \eta_{ijk} \, e_k \ ,
\end{align}
for $0<i,j,k\leqslant 7$, where $\eta_{ijk}$ are the octonionic structure constants which are completely skew-symmetric and whose only non-vanishing values up to permutation of indices are
\begin{align} \label{eq:etadef}
\eta_{123} = \eta_{145} = \eta_{176} = \eta_{246} = \eta_{257} = \eta_{347} = \eta_{365} = 1 \ .
\end{align}
It is naturally a $*$-algebra with $*$-involution given by
\begin{align}
e_0^* = e_0 \qquad \text{and} \qquad e_i^* = -e_i \ ,
\end{align}
for $0<i\leqslant 7$. 

In the following we write
\begin{align}
\CN_0 = \big\{(123)\,,\,(145)\,,\,(176)\,,\,(246)\,,\,(257)\,,\,(347)\,,\,(365)\big\}
\end{align}
for the set of triples of integers appearing in \eqref{eq:etadef}, and
\begin{align}
\CN_\pm = \big\{\big(\sigma(i)\,\sigma(j)\,\sigma(k)\big) \ \big| \ (ijk)\in\CN_0 \ , \ \sigma\in S_3 \ , \ {\rm sgn}(\sigma)=\pm 1\big\}
\end{align}
for the set of all even/odd permutations of these triples. Then
\begin{align}
\CN = \CN_+\cup\CN_-
\end{align}
is the set of all triples of distinct non-zero integers appearing as labels in the definition of the octonion multiplication table.

\subsection*{\underline{\sf Enveloping algebra}}

The structure of the enveloping algebra $\CCU(\FO)$ follows from the basic left module actions in the regular birepresentation
\begin{align} \label{eq:octmodule}
\widehat e_i\,\triangleright\, e_0 = \widehat e_0\,\triangleright\, e_i = e_i \qquad \text{and} \qquad \widehat e_i\,\triangleright\, e_j = -\delta_{ij} + \eta_{ijk}\,e_k  \ ,
\end{align}
for $i,j\in\{1,\dots,7\}$. This defines an $8{\times}8$ matrix representation $\varPi:\FO\longrightarrow\MM_8(\FC)$ of the octonions by left multiplication on themselves; in this representation we regard the octonion units $e_\mu$ for $\mu\in\{0,1,\dots,7\}$ as the standard unit vectors in $\FC^8$. The composition product is given by matrix multiplication
\begin{align}
\widehat e_\mu\circ\widehat e_\nu = \varPi(e_\mu)\cdot\varPi(e_\nu) \ ,
\end{align}
while the $\ast$-structure is represented by Hermitean conjugation ${}^\dag$. 

In the following we use the notation 
\begin{align}
\unit=\varPi(e_0) \qquad \text{and} \qquad E_i=\varPi(e_i)
\end{align}
for the images of the octonion generators in this representation. The representation by right multiplication is then obtained from this representation by Hermitean conjugation $E_i'=E_i^\dag$.

It is useful to realise the enveloping algebra $\CCU(\FO)$ as a subalgebra of the group of $8{\times}8$ signed permutation matrices~\cite{Kerber71}. Abstractly the latter is the \emph{signed symmetric group} $\RZ_2\,\wr\, S_8$ of degree eight, i.e.~the external wreath product of the cyclic group $\RZ_2=\{\pm1\}$ of order two with the symmetric group $S_8$, acting naturally by permutations $\sigma$ of the set $\{-8,\dots,-1,1,\dots,8\}$ such that $\sigma(-i)=-\sigma(i)$ for all $i$. We write $(j\,k)_-$ for the signed two-cycle whose only non-trivial actions are given by $(j\,k)_-(\pm\,j)=\mp\,k$ and $(j\,k)_-(\pm\,k)=\mp\,j$.

The octonions act on themselves via left multiplication as signed permutations, which for each $i\in\{1,\dots,7\}$ can be expressed in cycle notation as
\begin{align} \label{eq:signed2cycles}
(0\,i)_-\,(j_1k_1)_-\,(j_2k_2)_-\,(j_3k_3)_- \ ,
\end{align}
where $(ij_lk_l)\in \CN_+$ for $l=1,2,3$. Here $(j_lk_l)_-$ is the signed two-cycle representing the octonionic products $e_i\,e_{j_l} = e_{k_l}$ and $e_i\,e_{k_l} = -e_{j_l}$, whereas $(0\,i)_-$ represents $e_i\,e_0 = e_i$ and $e_i^2=-e_0$. 

The $8{\times}8$ matrices $E_i$ are all signed permutation matrices, which in the cycle decomposition \eqref{eq:signed2cycles} have non-zero matrix elements $(E_i)_{i\,0}=1=-(E_i)_{0\,i}$ and $(E_i)_{k_lj_l}=1=-(E_i)_{j_lk_l}$ for $l=1,2,3$.
They are skew-symmetric and traceless, and in fact also orthogonal. Similarly, the matrices $E_i\cdot E_j$ are again signed permutations that can each be expressed as a product of four distinct signed two-cycles as in \eqref{eq:signed2cycles}.

\begin{example}
The signed permutations generated by left multiplication with the imaginary unit octonions $e_1$ and $e_2$ are respectively given by
\begin{align}
(0\,1)_-\,(2\,3)_-\,(4\,5)_-\,(7\,6)_- \qquad \text{and} \qquad (0\,2)_-\,(3\,1)_-\,(4\,6)_-\,(5\,7)_- \ .
\end{align}
The corresponding signed permutation matrices are
\begin{align}
E_1 = {\scriptstyle\begin{pmatrix}
0 & -1 & 0 & 0 & 0 & 0 & 0 & 0 \\
1 & 0 & 0 & 0 & 0 & 0 & 0 & 0 \\
0 & 0 & 0 & -1 & 0 & 0  & 0 & 0 \\
0 & 0 & 1 & 0 & 0 & 0 & 0 & 0 \\
0 & 0 & 0 & 0 & 0 & -1 & 0 & 0 \\
0 & 0 & 0 & 0 & 1 & 0 & 0 & 0 \\
0 & 0 & 0 & 0 & 0 & 0 & 0 & 1 \\
0 & 0 & 0 & 0 & 0 & 0 & -1 & 0
\end{pmatrix}} 
\quad \text{and} \quad
E_2 = {\scriptstyle\begin{pmatrix}
0 & 0 & -1 & 0 & 0 & 0 & 0 & 0 \\
0 & 0 & 0 & 1 & 0 & 0 & 0 & 0 \\
1 & 0 & 0 & 0 & 0 & 0  & 0 & 0 \\
0 & -1 & 0 & 0 & 0 & 0 & 0 & 0 \\
0 & 0 & 0 & 0 & 0 & 0 & -1 & 0 \\
0 & 0 & 0 & 0 & 0 & 0 & 0 & -1 \\
0 & 0 & 0 & 0 & 1 & 0 & 0 & 0 \\
0 & 0 & 0 & 0 & 0 & 1 & 0 & 0
\end{pmatrix}}  \ .
\end{align}
\end{example}

For $(ijk)\in \CN$, define the \emph{phase matrix}
\begin{align}
P_{ijk} := E_i\cdot E_j\cdot E_k \ .
\end{align}
This is a traceless diagonal matrix with entries $-1$ at positions $0$, $i$, $j$ and $k$, and $+1$ at the remaining four positions. This simply expresses the fact that any triad of octonions indexed by the set $\CN$, together with the unit element, generate a subalgebra of $\FO$ isomorphic to the quaternion algebra $\FH$, and that the octonions can be obtained from the quaternions via the Cayley-Dickson construction. The $-1$ entries of $P_{ijk}$ represent the fact that the product of the three imaginary unit quaternions is equal to minus the unit element, whereas the $+1$ entries represent the orthogonal space to $\FH$ in $\FO$. There is similarly a product expression $E_1\cdot E_2\cdot \, \cdots \, \cdot E_7=-\unit$.

This characterizes the associative enveloping algebra $\CCU(\FO)$ as the unital subalgebra of $\MM_8(\FC)$ cut out by the relations
\begin{align}
E_i\cdot E_j = -\delta_{ij}\,\unit - \eta_{ijk}\,P_{ijk}\cdot E_k \quad \text{(no sum on $j,k$)} \ ,
\end{align}
where $(ijk)\in \CN$. The original octonion algebra is recovered by acting on $e_0$, i.e. via projection onto the zeroth column: $E_j\cdot e_0 = e_j$, which then yields the relations \eqref{eq:octrels}.

One consequence of these relations is the well-known result that left multiplication of the octonions on themselves generates (the complexifications of) the Lie algebra $\mathfrak{so}(8)$ as well as the Clifford algebra ${\sf Cl}(8)$ in eight dimensions. The permutation matrices $E_i$ form seven of the $28$ generators of $\mathfrak{so}(8)$, while the remaining $21$ generators are the permutation matrices $E_i\cdot E_j=-E_j\cdot E_i$ with $i\neq j$. We may also deduce the equation
\begin{align}
\bar E_\mu\cdot E_\nu + \bar E_\nu\cdot E_\mu = 2\,\delta_{\mu\nu}\,\unit
\end{align}
for $\mu,\nu\in\{0,1,\dots,7\}$, where $\bar E_0=E_0$ and $\bar E_i = -E_i$. Then the $16{\times}16$ matrices
\begin{align}
\varGamma_\mu = \begin{pmatrix}
0 & E_\mu \\ \bar E_\mu & 0
\end{pmatrix}
\end{align}
generate the Clifford algebra ${\sf Cl}(8)$. 

Similarly to Proposition~\ref{prop:Jordancompalg}, the octonionic descriptions of $\mathfrak{so}(8)$ and ${\sf Cl}(8)$ in fact extend to

\begin{proposition} \label{prop:octonionM8}
The octonion enveloping algebra generates all complex $8{\times}8$ matrices:
\begin{align}
\CCU(\mathbbm{O}) = \mathbbm{M}_8(\FC) \ .
\end{align}
\end{proposition}

\begin{proof}
While the only projectors in a division algebra such as $\FO$ are trivial, it is not difficult to construct  projectors of rank one in the enveloping algebra $\CCU(\FO)$ onto each of the unit octonions from $\unit$ together with sums and products of the phase matrices $P_{ijk}$ for $(ijk)\in \CN$. Multiplying these projectors by $E_i$ then gives the basis of $8{\times}8$ matrix units for $\MM_8(\FC)$.
\end{proof}

\begin{remark}
In quantum mechanics, the projectors constructed in the proof of Proposition~\ref{prop:octonionM8} can be used for projective measurements, while $\ii\,E_i$ provide nice examples of observables.
\end{remark}

\subsection*{\underline{\sf Tracial GNS representation}}

From \eqref{eq:octrels} it easy to see that there is a unique normalized $2$-cyclic trace which is defined on generators by
\begin{align}
\tau(e_0) = 1 \qquad \text{and} \qquad \tau(e_i) = 0 \ ,
\end{align}
for $0<i\leqslant 7$. It is straightforward to check that this is positive and $3$-cyclic as well.
The corresponding ideal of zero norm elements is again trivial:
\begin{align}
\CJ_\tau = \big\{a\in\mathbbm{O} \ \big| \ |a|^2=0\big\} = \{0\} \ ,
\end{align}
where we write $a=a_\mu\,e^\mu\in\mathbbm{O}$ with $a_\mu\in\FC$ for $\mu\in\{0,1,\dots,7\}$, and $|\,\cdot\,|$ is the standard norm on~$\FC^8$. 

Hence the tracial GNS Hilbert space can be identified with $\mathbbm{O}\simeq\FC^8$ as a vector space:
\begin{align}
\CH_\tau = \FC^8 \ ,
\end{align}
equipped with the standard inner product $(\,\cdot\,,\,\cdot\,)_\tau$. From the basic module actions \eqref{eq:octmodule} we can compute the $8{\times}8$ matrices generating the tracial GNS representation of the enveloping algebra as
\begin{align}
\pi_\tau(\widehat e_0) = \unit \qquad \text{and} \qquad \pi_\tau(\widehat e_i) = E_i \ ,
\end{align}
where $i\in\{1,\dots,7\}$, with
\begin{align}
\tau(A)=\pi_\tau(A)_{00}
\end{align}
for all $A\in\CCU(\FO)=\MM_8(\FC)$.

\subsection*{\underline{\sf Tracial uncertainties}}

Let us now look at a purely nonassociative phenomenon from the octonion algebra. Consider the unital subalgebra of the enveloping algebra $\CCU(\FO)$ generated by $\unit=\widehat e_0$ together with 
\begin{align}
A=\ii\,\widehat e_7 \qquad \text{and} \qquad B = \widehat e_1\circ \widehat e_2\circ \widehat e_4 \ .
\end{align}
These are observables: $A=A^*$ and $B=B^*$, which square to the unit element: $A\circ A=B\circ B=\unit$.

Consider the element $E$ of this subalgebra defined by
\begin{align}
A\circ B =: \ii\,E =-B\circ A \ .
\end{align}
It is also an observable which squares to the unit element: $E=E^*$ and $E\circ E=\unit$. It follows that
\begin{align}
[A,B]_\circ = 2\,\ii\,E \ .
\end{align}
The observable $E$ is normalized, $\tau(E)=1$, and hence
\begin{align} \label{eq:nonzerotrace}
\tau([A,B]_\circ) = 2\,\ii\,\tau(E) = 2\,\ii \neq 0 \ .
\end{align}
Thus we find that the trace of a finite-dimensional operator commutator is non-zero, providing an explicit realisation of the loss of cyclicity discussed in Remark~\ref{rem:losscyclicity}, which is not possible in the associative setting.

By Lemma~\ref{lem:uncert}, the corresponding uncertainty relation in the tracial state reads
\begin{align}
\Delta_\tau A \ \Delta_\tau B \geqslant \tfrac12\,\big|\big\langle[A,B]_\circ\big\rangle_\tau\big| = 1 \ .
\end{align}
We can compute the uncertainties for measurement of the observables $A$ and $B$ in the tracial state using $\tau(A)=\tau(B)=0$ to get
\begin{align}
\Delta_\tau A=\sqrt{\tau(A\circ A)} = 1 \qquad \text{and} \qquad \Delta_\tau B=\sqrt{\tau(B\circ B)} = 1 \ .
\end{align}
It follows that the tracial state is a minimum uncertainty state, which would again be impossible in any associative treatment of quantum mechanics. The triple $(A,B,E)$ can be interpreted as providing single qubit nonassociative quantum gates, analogous to the Pauli $(X,Y,Z)$ gates of traditional quantum mechanics, except that here they act in a reducible eight-dimensional representation with $\tau(E)\neq0$.

\begin{remark}
These considerations nicely illustrate the difference between nonassociative quantum mechanics with octonions and ordinary quantum mechanics with an eight-dimensional Hilbert space. In the latter case the $8{\times}8$ matrix point of view would ignore the nonassociative operator origin and thus any special meaning to the original operators. Moreover, the states are quite different in the two cases. In particular, in eight-dimensional quantum mechanics the tracial state would be the normalized matrix trace, whereas for our formulation of nonassociative quantum mechanics based on the octonions it is the projection onto the unit element $e_0$ (i.e. the $(00)$-matrix element).
\end{remark}

\section{Conclusions}
\label{sec:conclusion}

In applications to quantum mechanics, elements of the nonassociative algebra $\CA$ play a dual role: as `state vectors' (tracial GNS construction,  \S\ref{sec:trace}), in which case we will often denote elements of $\CA$ suggestively as $\psi\in\CA$, and also as basic nonassociative observables as in the examples of \S\ref{sec:Intro}. On the other hand, elements of the enveloping algebra $\CCU(\CA)$ play the role of `operators', including more general observables. Since $\CA$ embeds into $\CCU(\CA)$ as a vector space, elements of $\CA$ also have an interpretation in terms of operators; this is the reason for the suggestive notation $\widehat a\in\CCU(\CA)$ for $a\in\CA$. It will be this latter picture that captures truly nonassociative features in our algebraic formulation of quantum theory: the nonassociative algebra $\CA$ together with additional data (see \S\ref{sec:trace}) completely define the theory.

The enveloping algebra $\CCU(\CA)$ setting is required for positive dynamics and a consistent probability interpretation as discussed in section \ref{sec:CPM}. (Previous alternative proposals like the one quoted in \eqref{eq:H1H2} can be tweaked to preserve reality and normalization, but not positivity, i.e.\ such alternative proposals  inadvertently lead to negative probabilities and are therefore not compatible with a quantum mechanical probability interpretation.) While $\CCU(\CA)$ is associative by construction, it does not remove interesting nonassociative effects: There are basic non-associative observables like e.g.\ \eqref{eq:Jacjm} and \eqref{eq:bonafide} that would be zero in an associative setting. Furthermore there are bonafide nonassociative effects like the non-vanishing trace of commutators in a finite-dimensional setting and the minimum uncertainty of the tracial state as described in \S\ref{sec:oct} , as well as effects like volume uncertainties and space coarse-graining that have been discussed in previous work. 
The methods developed in this work can for instance be used to study the quantum mechanics of charged particles in the presence of magnetic sources, non-geometric flux backgrounds in string theory, and even more exotic quantum algebras or operator systems.

\subsection*{\underline{\sf Acknowledgements}}

We thank Martin Cederwall and Bernd Henschenmacher for helpful discussions and correspondence. We are grateful to the Erwin Schr\"odinger International Institute for Mathematics and Physics (ESI), the Mainz Institute for Theoretical Physics (MITP) of the DFG Cluster of Excellence PRISMA${}^+$ (Project ID 39083149), the Ru\dj er Bo\v{s}kovi\'c Institute (RBI),  The Galileo Galilei Institute for Theoretical Physics (GGI), the Department of Mathematical Sciences of Durham University, and the Mathematical Institute of Charles University for their hospitality and support during  various stages of this work. This article is based upon work from COST Actions CaLISTA CA21109 and THEORY-CHALLENGES CA22113 supported by COST (European Cooperation in Science and Technology). 

\subsection*{\underline{\sf Statements and Declarations}}

The authors have no competing interests to declare that are relevant to the content of this article. Data sharing is not applicable to this article as no datasets were generated or analyzed during the current study.

\end{document}